\documentclass[10pt, twocolumn, comsoc]{IEEEtran}
\usepackage{graphicx,epsfig}
\usepackage[noadjust]{cite}
\usepackage{mcite}
\usepackage{amsfonts,helvet}
\usepackage{fancyhdr}
\usepackage{threeparttable}
\usepackage{epsf,epsfig}
\usepackage{amsthm}
\usepackage{amsmath}
\usepackage{boldline}
\usepackage{booktabs}
\usepackage{siunitx}
\usepackage{amssymb}
\usepackage{lipsum}
\usepackage{epstopdf}

\usepackage{subcaption}
\usepackage{gensymb}

\usepackage{dsfont}
\usepackage{color}
\usepackage{array}
\usepackage{algpseudocode}
\usepackage{algcompatible}
\usepackage{enumerate}
\usepackage{gensymb}
\usepackage{cancel}

\usepackage[linesnumbered,ruled,noend]{algorithm2e}

\usepackage{siunitx}
\usepackage{fancyhdr}
\usepackage{graphicx,epsfig}
\usepackage{wrapfig}
\usepackage{ragged2e}
\usepackage{algpseudocode}
\usepackage{algcompatible}
\usepackage{threeparttable}
\allowdisplaybreaks

\newtheorem{theorem}{Theorem}

\newtheorem{corollary}{Corollary}

\newtheorem{lemma}{Lemma}

\usepackage{eucal}

\begin{document}

\title{  Modeling and Coverage Analysis of $K$-Tier Integrated Satellite-Terrestrial Downlink Networks}

\author{Jungbin Yim,~\IEEEmembership{Student Member,~IEEE,} Jeonghun Park,~\IEEEmembership{Member,~IEEE,} and Namyoon Lee,~\IEEEmembership{Senior Member,~IEEE}, 



\thanks{J. Yim is with School of Electrical Engineering, POSTECH, South Korea (email:{\texttt{jungbinyim@postech.ac.kr}}), J. Park is with School of Electrical and Electronic Engineering, Yonsei University, South Korea (e-mail:{\texttt{jhpark@yonsei.ac.kr}}), N. Lee is with Department of Electrical Engineering, Korea University, South Korea (e-mail:{\texttt{namyoon@korea.ac.kr}}).  
}}

\maketitle

\begin{abstract}
Integrated satellite-terrestrial networks (ISTNs) can significantly expand network coverage while diminishing reliance on terrestrial infrastructure. Despite the enticing potential of ISTNs, there is no comprehensive mathematical performance analysis framework for these emerging networks. In this paper, we introduce a tractable approach to analyze the downlink coverage performance of multi-tier ISTNs, where each network tier operates with orthogonal frequency bands. The proposed approach is to model the spatial distribution of cellular and satellite base stations using homogeneous Poisson point processes arranged on concentric spheres with varying radii. Central to our analysis is a displacement principle that transforms base station locations on different spheres into projected rings while preserving the distance distribution to the typical user. By incorporating the effects of Shadowed-Rician fading on satellite channels and employing orthogonal frequency bands, we derive analytical expressions for coverage in the integrated networks while keeping full generality. Our primary discovery is that network performance reaches its maximum when selecting the optimal density ratio of users associated with the network according to the density and the channel parameters of each network. Through simulations, we validate the precision of our derived expressions.
\end{abstract}


\section{Introduction}

 Satellite communication and networking in low earth orbit (LEO) is experiencing rapid growth, presenting substantial opportunities for facilitating novel and inventive applications. It achieves this by offering high-speed Internet connectivity with low delays on a global scale. The extensive deployment of LEO satellites by major industry players like Starlink and OneWeb has made LEO satellite network access widely accessible around the world \cite{del2019technical}.  These satellite networks complement traditional terrestrial networks, especially in areas where terrestrial coverage is limited or cost-prohibitive \cite{guidotti2017satellite,cioni2018satellite,guidotti2019architectures,lin2021path}. There has recently been significant interest in integrating cellular and satellite networks to attain synergistic gains in expanded coverage and transmission rate enhancement. For instance, the integrated satellite-terrestrial networks (ISTNs), as explored in the 3rd generation partnership project (3GPP), offer the potential for seamless global broadband access and ensure continuous service, making it a prominent development in next-generation communication networks, as in \cite{yang20196g,zhang20196g,qian2020integrated,qian2021multi,5gamericas2022NTN,al2022next}.




Characterizing coverage and transmission rate enhancements in ISTNs is crucial for optimal network design. Despite its significance, establishing these improvements with respect to crucial network parameters, including network densities, channel fading effects, and spectrum availability, remains an open challenge. While system-level simulations are essential approaches for ISTN analysis and design, there is a growing imperative for a complementary analytical approach to facilitate benchmarking and comparison. Such analysis sheds light on critical dependencies within the system and offers guidance on network optimization. Stochastic geometry serves as a mathematical tool to offer manageable expressions for coverage in various wireless network types, including ad-hoc \cite{baccelli2006aloha}, device-to-device \cite{lee2014power,al2016stochastic,chun2017stochastic}, cellular \cite{andrews2011tractable,dhillon2012modeling,jo2012heterogeneous,huang2012analytical,lee2014spectral,bai2014coverage,di2015stochastic,park2016optimal,park2018inter}, {\color{black} aerial \cite{wang2023resident}}, and satellite networks \cite{talgat2020stochastic,okati2020downlink,okati2022nonhomogeneous,al2021analytic,al2021optimal,na2021performance,park2022tractable}.

Continuing the success of utilizing stochastic geometry, we aim to introduce an analytical framework for characterizing downlink coverage expression in ISTNs, particularly when utilizing orthogonal frequency bands across networks. Through this analytical framework, we  provide insights into network design principles that facilitate the realization of synergistic advantages within ISTNs.

\vspace{-0.2cm}


\subsection{Prior Works}

A popular approach to modeling cellular and satellite networks is deterministic network topology. The two-dimensional hexagonal grid model is one of the most commonly utilized and accepted cellular network models. In addition, this grid model serves as the foundation for system-level simulations, even though performing detailed analytical analyses is generally challenging \cite{gilhousen1991capacity,viterbi1994other}. However, it is essential to note that both the scalability and accuracy of the grid model can be questionable in scenarios characterized by network heterogeneity.

With the advent of heterogeneous cellular networks (HCNs), which are large and difficult to know or predict the locations and channels of all nodes, the amount of uncertainty has increased significantly, making analysis difficult \cite{haenggi2009stochastic}. Therefore, a less conventional approach involves permitting the base station (BS) locations to be determined by a stochastic point process, {\color{black} as referenced in \cite{baccelli2006aloha,lee2014power,al2016stochastic,chun2017stochastic,andrews2011tractable,dhillon2012modeling,jo2012heterogeneous,huang2012analytical,lee2014spectral,bai2014coverage,di2015stochastic,park2016optimal,park2018inter,wang2023resident,talgat2020stochastic,okati2020downlink,okati2022nonhomogeneous,al2021analytic,al2021optimal,na2021performance,park2022tractable}}. While this model may appear logical for femtocells, which are deployed in unpredictable and unplanned positions, it might raise skepticism when applied to the higher network tiers that are typically centrally orchestrated. However, as depicted in \cite{andrews2011tractable}, the disparity between randomly chosen locations and planned positions may not be as substantial as initially anticipated, even in the case of tier 1 macro BSs. Also, the results in \cite{andrews2011tractable} have demonstrated that in a single-tier network, where BS locations are drawn from a Poisson Point Process (PPP), this resulting network model is nearly as accurate as the standard grid model when compared to a real 4G network. Importantly, this model allows the application of valuable mathematical tools derived from stochastic geometry \cite{chiu2013stochastic}, enabling the development of a manageable analytical model that remains precise.

Recently, like terrestrial networks, satellite networks are also gaining momentum in tractable analysis using stochastic geometry \cite{talgat2020stochastic,okati2020downlink,okati2022nonhomogeneous,al2021analytic,al2021optimal,park2022tractable}. {\color{black} The motivation of applying stochastic geometry to analyze network coverage in LEO mega-constellations where it is difficult to predict the locations and channels of all nodes is reasonable \cite{wang2022ultra}. To support this motivation, the authors in \cite{wang2022evaluating} attempted to verify the reasonableness of random modeling by investigating the accuracy of modeling based on stochastic geometry for satellite networks}. There are two major point processes used in satellite networks: binomial point process (BPP) \cite{talgat2020stochastic,okati2020downlink,okati2022nonhomogeneous,na2021performance} and PPP \cite{al2021analytic,al2021optimal,na2021performance,park2022tractable}. BPP is a special case of PPP because the number of satellites is fixed. On the other hand, in PPP, the number of satellites is also determined by probability, which is an element of additional randomness. Comparing the results of \cite{okati2020downlink} and \cite{park2022tractable}, the results of PPP are more tractable.

As such, coverage analysis of terrestrial and satellite networks using stochastic geometry has been studied extensively, but integrating the two networks is challenging due to the heterogeneity in the network geometries. To be specific, cellular networks are typically represented in a two-dimensional plane, whereas satellite networks are modeled on a spherical shell. Most previous studies have ignored the distinction in homogeneous network geometries. Instead, they have used an approximation approach, in which the locations of the non-terrestrial platforms are modeled with a PPP in a two-dimensional plane with proper density matching \cite{song2022cooperative,al2021modeling}. Very recently, a unified network model for ISTNs using the sphere shells has been proposed in \cite{park2023unified}. However, in \cite{park2023unified}, the tractability of the network model is reduced by randomly determining the height of the terrestrial BSs and the altitude of satellites. Furthermore, although terrestrial and satellite networks are assumed to share the same spectrum, this may differ from the actual network implementation.

\vspace{-0.2cm}


\subsection{Contributions} 
The main contributions are summarized as follows:
\begin{itemize}
    \item {\bf Integrated Network Modeling:} We introduce a unified $K$-tier ISTN model, in which each class of satellite and cellular BSs is drawn from a homogeneous PPP on concentric spheres with distinct radii per tier, i.e., network heights per tier. This model extends a single-layered satellite network model in \cite{park2022tractable} into the multiple-layered network. In addition, the modeling of the locations of cellular BSs on the surface of the Earth with a fixed height is different from the time-honored model of the cellular network on the two-dimensional plane.  \cite{andrews2011tractable}. Further, the proposed model has the key distinction from our recent work \cite{park2023unified} in that the BSs, being different layers, do not interfere with users associated with that specific network. This scenario captures a practical case in which satellite and cellular network operators do not share the frequency spectrum. At the same time, a user can be connected by a cellular or satellite network in one of multiple tiers opportunistically, i.e., multi-connectivity service. 
    

\item {\bf Coverage Distribution for Multi-Connectivity:} We characterize the coverage distribution for a typical user connected to the BS that provides the maximum instantaneous SINR in the $K$-tier network. Our primary technical method in characterizing the distribution is to harness a random displacement that transforms a PPP distributed on the surface of a sphere into another PPP distributed on a ring in the plane. This transformation maintains a statistically invariant distance distribution from a typical user to the BSs, provided that the densities of the PPPs are appropriately scaled. The coverage distribution we derive accounts for the shadowed-Rician fading effects in satellite links and considers the influence of multi-tier network densities. In specific plausible scenarios, we present the resulting expression in closed form, providing valuable insights for system design.

\item {\bf System Design Insights:} Our first key observation is that the SIR coverage of the cellular network is a variant of network densities. This is a different observation in \cite{andrews2011tractable}; the SIR distributions are invariant for the network density when modeling the BS locations according to a homogeneous PPP on $\mathbb{R}^2$. This disparity arises because our proposed model can capture the visibility probability of a typical downlink user within the cellular network. Another critical observation is the macro-diversity that can be achieved as tiers increase. Because the tiers do not share spectrum due to multi-band, the probability of outage decreases without the influence of interference as the number of tiers increases. This macro-diversity varies depending on the channel fading, the density of BSs for each tier, and the density ratio of users associated with the tier. Therefore, there is an optimal density ratio of users associated with the tier that can achieve maximum macro-diversity depending on the channel fading and the density of BSs for each tier. We adopt a flexible cell association that controls the proportion of users assigned to the network through a biasing factor. Using this biasing factor, we can examine the optimal association ratio for users to maximize the coverage performance according to the channel fading and the density. This examination provides a comprehensive network design guideline for the seamless integration of these two networks, ensuring optimal performance for multi-connectivity services.
\end{itemize}

\section{System Model} \label{sec:system model}
In this section, we present an ISTN network and channel fading models for downlink communications. Subsequently, we describe a cell association strategy and provide a performance metric for analyzing downlink ISTNs.
\vspace{-0.23cm}

\subsection{Network Model}
We consider a $K$-tier ISTN consisting of $K_{\sf T}$-tier cellular and $K_{\sf S}$-tier satellite networks, where $K=K_{\sf T}+K_{\sf S}$. For the sake of notational convenience, we represent the indices corresponding to the terrestrial networks as $k=1,...,K_{\sf T}$, and those for the satellite network as $k=K_{\sf T}+1,...,K$.

{\bf Spatial distribution of the BSs and users:} To explain the spatial network geometry, we define the surface of the sphere with radius $R_k$ as
\begin{align}
    {\mathcal S}_{k}=\{{\bf x}_k\in\mathbb{R}^3:\lVert\mathbf{x}\rVert_2=R_k\}. \label{eq:surface of sphere}
\end{align}
In our network model, we assume that the locations of base stations (BSs) for cellular and satellite networks are placed on the surface of sphere ${\mathcal S}_{k}^2$ with radius $R_k$ for $k\in [K]$ as illustrated in Fig. \ref{fig:Spherical cap}.  We also denote the locations of the $i$th BS for the $k$th tier network by ${\bf x}_{k,i}\in {\mathcal S}_{k}^2$. The locations of the BSs for the $k$th tier network $\Phi_k=\{{\bf x}_{k,1},...,{\bf x}_{k,N_k}\}$ are assumed to be distributed according to homogeneous PPPs with density $\lambda_k$. Therefore, the number of BSs on ${\mathcal S}_{k}^2$, $N_k$, follows a Poisson distribution with mean $\lambda_k|{\mathcal S}_{k}|=4\pi \lambda_kR_{k}^2$.

Downlink users are assumed to be located on the surface of the Earth with radius $R_{\sf E}$, i.e., $\mathcal{S}_{R_{\sf E}}$. The locations of the users are distributed according to a homogeneous PPP $\Phi_U=\{\mathbf{u}_1,...,\mathbf{u}_{N_U}\}$ with density $\lambda_U$ on $\mathcal{S}_{R_{\sf E}}$, where $\mathbf{u}_i$ is the location of the $i$th user and $N_U$ is the number of users that follows a Poisson distribution with mean $\lambda_U|\mathcal{S}_{R_{\sf E}}|=4\lambda_U\pi R_{\sf E}^2$. 

Our model differs from the classical cellular network model, in which the locations of the BSs and users are distributed in a two-dimensional plane $\mathbb{R}^2$. Whereas, in our model, the BSs and users are located on the surface of the spheres with different radii, which incorporates the curvature of the Earth.

{\bf Typical spherical cap and visible region:} Under the premise that $\Phi_k$ and  $\Phi_U$  follow homogeneous PPPs on the concentric spheres $\mathcal{S}_k$ for $k\in[K]$ and $\mathcal{S}_{\sf E}$, the statistical distribution of $\Phi_k$ with respect to any element in $\Phi_U$ remains invariant under rotation. Leveraging Slivnyak's theorem [34], we can, without loss of generality, focus on a typical user positioned at $(0,0,R_{\sf E})$. As  depicted in Fig. \ref{fig:Spherical cap}, we denote the typical spherical cap of the $k$th network by  $\mathcal{A}_k \subset \mathcal{S}_k$.
This typical spherical cap is defined as the field of view at the typical receiver's location $(0,0,R_{\sf E})$ with the visible elevation angle $\theta_{k}$ as
\begin{equation}
    \theta_{k}=\left(0,\theta_{k}^{\rm max}\right],
\end{equation}
where $\theta_k^{\rm{max}}=\cos^{-1}\left(\frac{R_{\sf E}}{R_k}\right)$ is the maximum visible elevation angle.


Without loss of generality, our attention is directed toward the downlink user performance within the confines of the typical spherical cap. It is essential to highlight that, exclusive to the typical user, only BSs positioned within the typical spherical cap $\mathcal{A}_k$ for $k\in [K]$ are observable. The spatial distributions of these visible cellular and satellite BSs are identical to the homogeneous PPPs with density $\lambda_k$. 

\vspace{-0.2cm}

\begin{figure}[t!]
  \centering
  \includegraphics[width=0.3\textwidth]{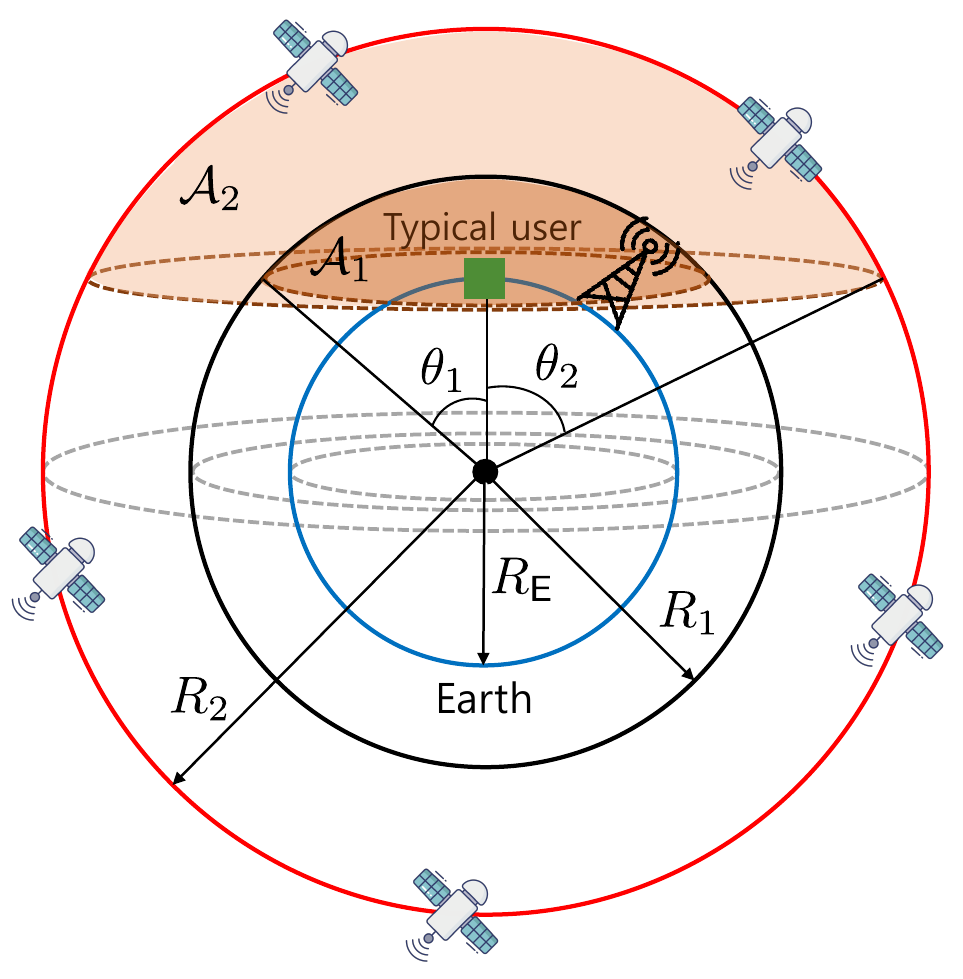}
  \caption{An illustration of the 2-tier ISTN network geometry when  $K_{\sf{T}}=K_{\sf{S}}=1$. The shaded area of the spheres represents the spherical cap, which is the area of BSs that can be associated with a typical user.} \vspace{-0.2cm}
  \label{fig:Spherical cap}
\end{figure}

\subsection{Channel and Antenna Gain Model}

The propagation through the wireless channel is modeled by the combination of path-loss attenuation, shadowing, and small-scale fading effects.

{\bf Path-loss model:} The path-loss model is dependent on the distance from the $i$th BS of the $k$th tier located at ${\bf x}_{k,i}\in\Phi_k$ for $k\in [K]$ to the typical user located at $\mathbf{u}_1$ as
\begin{equation}
    \lVert{\bf x}_{k,i}-\mathbf{u}_1\rVert^{-\alpha_k},
\end{equation} 
where $\alpha_k\ge 2$ is the path loss exponent. We assume a higher path loss exponent for cellular networks than satellite networks due to their distinct channel propagation characteristics. Cellular networks typically navigate urban communication scenarios characterized by rich scattering, necessitating a higher path loss exponent. In contrast, satellite networks predominantly operate in line-of-sight (LOS) scenarios, prompting a lower path loss exponent. Notably, the flexibility to adjust these path loss exponents independently for each tier enhances adaptability to diverse network conditions.


{\bf Antenna gain model:}  We assume that all BSs in the cellular and satellite networks adopt directional beamforming. For analysis tractability, we use a simple two-level antenna gain model wherein the main lobe is confined within the beamwidth, and the side lobe provides a reduced beam gain compared to the main lobe. We denote the main and side lobe antenna gains for the BSs in the $k$th tier network by ${G}_k^{\sf M}$ and ${G}_k^{\sf S}$. Assuming that BS $i$ in the $k$th tier network is catering to the downlink user through main lobe beam alignment, while the beams of the interfering links are oriented towards the side lobe angles, the modeled effective beam gain is expressed as follows:
\begin{align}
    G_{k,i}=\begin{cases}
			G_{k,i}^{\sf M}G^{\sf M}\frac{c^2}{(4\pi f_k)^2}, & \text{ BS $i$ is serving }\\
          G_{k,i}^{\sf S}G^{\sf S}\frac{c^2}{(4\pi f_k)^2}, & \text{otherwise},
		 \end{cases}
\end{align}
 where $G^{\sf M}$ and $G^{\sf S}$ denote the main and side lobe beam gains of the downlink users. This simple two-lobe antenna gain model allows the analysis to be more tractable, and its effectiveness becomes particularly pronounced when coupled with sophisticated beam pattern techniques, such as Dolph-Chebyshev beamforming weights \cite{dolph1946current}.

{\bf Fading model:} For satellite links, we use the shadowed-Rician fading model, a time-honored fading model that captures land mobile satellite (LMS) links. Let $H_{k,i}$ be the fading power from the $i$th BS of the $k$th tier to the typical user. The probability density function (PDF) of $H_{k,i}$ is given by \cite{abdi2003new}
\begin{equation}
    f_{H_{k,i}}(x)=Z_k e^{-\beta_k x} e^{\delta_k x}\sum_{l=0}^{m_k-1}\frac{(-\delta_k x)^l(1-m_k)_l}{(l!)^2},
\end{equation}
 where $(w)_l=w(w+1)\times\cdots\times(w+l-1)$ is the Pochhammer symbol,  $Z_k=\left(\frac{2b_k m_k}{2b_k m_k+\Omega_k}\right)^{m_k}/(2b_k)$, $\beta_k=\frac{1}{2b_k}$, and $\delta_k=\frac{\Omega_k}{2b_k(2b_k m_k+\Omega_k)}$. Here, $m_k \in \mathbb{Z}^+$ denotes the shape parameter of the Nakagami-$m$ distribution. Also, $\Omega_k$ and $2b_k$ are the average power of the LoS and multi-path components of the $k$th tier downlink channel, respectively. Then, the cumulative distribution function (CDF) of $H_{k,i}$ is given by
\begin{equation}
    F_{H_{k,i}}(x)=\sum_{l=0}^{m_k-1}\frac{Z_k(-\delta_k)^l(1-m_k)_l}{(l!)^2(\beta_k-\delta_k)^{l+1}}\gamma(l+1,(\beta_k-\delta_k)x), \label{eq:CDF of SRF}
\end{equation}
where $\gamma(s,x) =\int_{0}^x e^{-t}t^{s-1}{\rm d}t$ is the lower incomplete gamma function. The shadowed-Rician fading model exhibits remarkable versatility. Notably, the Nakagami-$m$ fading emerges as a special case of shadowed-Rician fading by setting the parameter $b_k$ to sufficiently small values. In our analysis, we leverage both the shadowed-Rician and Nakagami-$m$ channel fading models for the satellite and terrestrial channels, respectively.

\vspace{-0.2cm}




\subsection{Flexible Network Association for Data Offloading}
We propose an open access policy for facilitating flexible network association, wherein a downlink user can establish connections with all $K$-tier networks, each equipped with orthogonal frequency bands. For a typical downlink user positioned at ${\bf u}_1$, first measures the effective ERP from the nearest BS within the $j$th tier, located at ${\bf x}_{j,1}$ with bias factor $B_j$ as
\begin{align}
    P_j^{\rm{eff}}=P_j G_{j,1} B_j E_j||{\bf x}_{j,1}-\mathbf{u}_1||^{-\alpha_j}, \label{eq:ERP}
\end{align}
where $P_j$ is the transmit power of the BS in the $j$th tier and $E_j$ denotes the average channel fading gain of the $j$th tier. The latter encompasses distinct values for shadowed-Rician satellite channels and the Nakagami-$m$ terrestrial channel. The bias factor plays a crucial role in influencing the dynamics of data offloading in multi-tier networks. Downlink users tend to be naturally inclined towards connecting to terrestrial networks without using the bias factor, owing to their inherently larger ERP than satellite networks. To serve more users, satellite network operators may adjust the bias factor upwards, enhancing the ERP and favoring satellite connections.

After measuring the ERP, the downlink user chooses the best network that produces the highest ERP as
\begin{align}
    k = \arg\max_{j\in [K]} P_j^{\rm{eff}}.
\end{align}
 
\vspace{-0.2cm}



\subsection{Performance Metric}
Let $J\in[K]$ be a random variable indicating the connected network index. Given that the typical downlink user is linked to the $k$th network, i.e., $J=k$, we first define the coverage probability. As the downlink user is served by the closest BS in the $k$th tier, placed at ${\bf x}_{k,1}$, the signal-to-interference-plus-noise ratio (SINR) of the typical downlink  user positioned at $\mathbf{u}_1$ is
\begin{align}
      {\rm{SINR}}_{k}=\frac{H_{k,1}\lVert{\bf x}_{k,1}-\mathbf{u}_1\rVert^{-\alpha_k}}{\sum_{i=2}^{N_k}\tilde{G}_k H_{k,i}\lVert{\bf x}_{k,i}-\mathbf{u}_1\rVert^{-\alpha_k}+\hat{\sigma}_k^2}, \label{eq:SINR simplified}
\end{align}
where $\hat{\sigma}_k^2=\frac{\sigma_k^2}{G_k P_k}$ is the normalized noise power of the $k$th tier and  $\tilde{G}_k=\frac{\bar{G}_k}{G_k}<1$ is the normalized antenna gain. Notice there is no cross-tier interference because all networks are assumed to use orthogonal frequency bands. The probability that the instantaneous SINR of the typical downlink  user from the $k$th tier is greater than a target SINR $T$ is defined as
\begin{align}
    {\sf P}_{k}^{\sf cov}(T) = \mathbb{P}[{\rm{SINR}}_k>T \mid J=k].
\end{align}
Since the typical user is associated with at most one tier, from the law
of total probability, the coverage probability of the $k$th tier network under the open access policy is given as
    \begin{align}
        {\sf P}^{\sf cov}(T)&= \mathbb{P}\left[\cup_{k=1}^{K}{\sf SINR}_k >T\right]\nonumber\\
        &=\sum_{k=1}^{K} \mathbb{P}[{\rm{SINR}}_k>T \mid J=k] \mathbb{P}[J=k] \nonumber \\
        &=\sum_{k=1}^{K}{\sf P}_{k}^{\sf cov}(T)\mathbb{P}[J=k], \label{eq:Pcov}
    \end{align}
where the network association probability is defined as
\begin{align}
    \mathbb{P}[J=k] =\mathbb{P}\left[ P_k^{\rm eff} > \max_{j\neq k} P_j^{\rm eff} \right].
\end{align}
This network association probability is chiefly determined by the nearest distribution of the BS in each tier, i.e., the density of the network $\lambda_k$ and the bias factors $B_k$ for $k\in [K]$. In this paper, we shall characterize the coverage probability in \eqref{eq:Pcov} and provide a system design guideline for optimizing the network parameters, such as the network densities and bias factors.

\vspace{-0.3cm}

\section{Preliminaries} \label{sec:mathematical preliminaries}
In this section, we introduce some useful lemmas that are essential for establishing the coverage probability expressions. 

\vspace{-0.25cm}

\subsection{Geometry Transform from a Spherical Cap to a Ring}
We introduce a displacement lemma that transforms the point process in a visible spherical cap into the point process in a projected ring. To explain this, we define a circular ring $\mathcal{\tilde A}_k\subset \mathbb{R}^2$ centered at $(0,0)\in \mathbb{R}^2$ with outer and inner radii, denoted as  $R_{{\rm{min}},k}=R_k-R_{\sf E}$ and $R_{{\rm{max}},k}=\sqrt{R_k^2+R_{\sf E}^2-2R_{\sf E} R_k\cos\theta_k}$ as illustrated in Fig. \ref{fig:Transformed spherical cap}. Let ${\tilde \Phi}$ be a homogeneous PPP on $\mathcal{\tilde A}$ with density ${\tilde \lambda}_k$.


\begin{figure}[t!]
  \centering
  \includegraphics[width=0.35\textwidth]{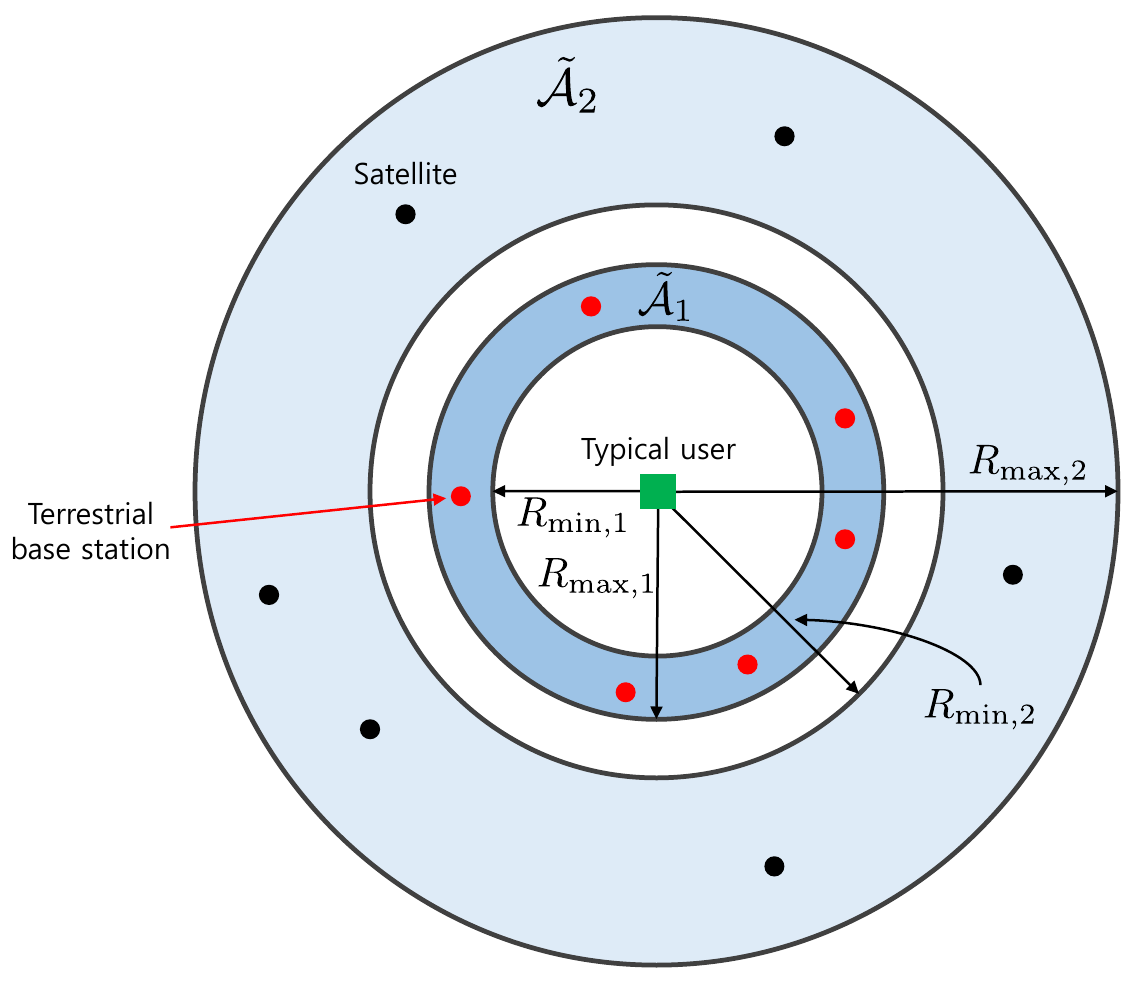}
  \caption{  An illustration of the transformed geometry for the 2-tier ISTN. The shaded area represents annulus $\tilde{\mathcal{A}}_k$ corresponding to spherical cap $\mathcal{A}_k$ for $k=1,2$.} \vspace{-0.2cm}
  \label{fig:Transformed spherical cap}
\end{figure}
\vspace{0.2cm}
\begin{lemma}
The distance distributions from a typical user are preserved between $\Phi_k$ on $\mathcal{A}_k  \subset \mathbb{R}^3$ with density $\lambda_k$ and  ${\tilde \Phi}_k$ on $\mathcal{\tilde A}_k \subset \mathbb{R}^2$ with density ${\tilde \lambda}_k$, provided that
\begin{equation}
    \tilde{\lambda}_k=\frac{R_k}{R_{\sf E}}\lambda_k.
\end{equation}
\end{lemma}

\begin{proof}
    See \cite{kim2023coverage}.
\end{proof}


Applying this geometry transformation lemma allows us to leverage more straightforward geometric intuition to analyze the coverage probability of sophisticated $K$-tier ISTNs.

\vspace{-0.2cm}


\subsection{First Touch Distribution}
The distribution of the nearest BS from the typical downlink user is a key ingredient in analyzing the coverage probability. The following lemma elucidates the first-touch distribution within the transformed geometry. Unlike the traditional cellular networks modeled in $\mathbb{R}^2$, our network model using the transformed geometry on the ring $\mathcal{\tilde A}_k$ requires computing the nearest BS's distance distribution conditioned that one BS at least exists in the ring. The void probability for ${\tilde \Phi}_k$ on $\tilde{\mathcal{A}}_k$ is 
\begin{align}
     \mathbb{P}[{\tilde \Phi}_k(\tilde{\mathcal{A}}_k)=0]=e^{-\pi\tilde{\lambda}_k\left(R_{{\rm{max}},k}^2-R_{{\rm{min}},k}^2\right)}.
\end{align}
From the void probability, the following lemmas describes the conditional first-touch distribution of the BS on ${\tilde \Phi}_k$.
\begin{lemma} \label{lem2}
    Let $r$ be the distance from a typical user to the nearest BS of the $k$th tier. Then, the PDF of $r$ conditioned on that at least one BS is visible is obtained as
    \begin{align}
        &f_{k}(r|{\tilde \Phi}_k(\tilde{\mathcal{A}}_k)>0) \nonumber \\
        &=\begin{cases}
            \frac{2r\tilde{\lambda}_k\pi e^{-\tilde{\lambda}_k\pi(r^2-R_{{\rm{min}},k}^2)}}{1-e^{-\pi\tilde{\lambda}_k\left(R_{{\rm{max}},k}^2-R_{{\rm{min}},k}^2\right)}},&{\rm{for}}\;R_{{{\rm{min}}},k}<r<R_{{{\rm{max}}},k},\\
            0,&{\rm{otherwise}}.
        \end{cases}
    \end{align}
    \label{conditional nearest distance distribution}
\end{lemma}

\begin{proof}
    See Appendix \ref{proof of conditional nearest distance distribution}.
\end{proof}
 
This PDF will be a key role in computing the coverage probability of the $k$th tier network.

\vspace{-0.2cm}

\subsection{Network Connection Probability}

By the open access policy, we assume that the typical downlink user is connected to the network generating the highest ERP in \eqref{eq:ERP}. The subsequent lemmas elucidate the probability of network connection contingent upon the maximum ERP under the open access policy.

\begin{lemma}
   The probability that a typical user is associated with the $k$th tier conditioned on that at least one BS is visible in the $k$th tier is
    \begin{align}
        \mathbb{P}&[J=k|{\tilde \Phi}_k(\tilde{\mathcal{A}}_k)>0] \nonumber \\
        &=\frac{2\pi\tilde{\lambda}_k}{1-e^{-\pi\tilde{\lambda}_k\left(R_{{\rm{max}},k}^2-R_{{\rm{min}},k}^2\right)}}\int_{R_{{\rm{min}},k}}^{R_{{\rm{max}},k}}r\prod_{j=1}^{K}{\sf P}_{k,j}{\rm{d}}r,
    \end{align}
    where
    \begin{align}
        &\!\!\!{\sf P}_{k,j}=\mathbb{P}\left[||{\bf x}_{j,1}-\mathbf{u}_{1}||>\left(\hat{P}_j\hat{G}_j\hat{B}_j\hat{E}_j\right)^{\frac{1}{\alpha_j}}r^{\frac{1}{\hat{\alpha_j}}}\big|{\tilde \Phi}_k(\tilde{\mathcal{A}}_k)>0\right] \nonumber \\
        &\!\!\!\!\!=\begin{cases}
            0, &U_{k,j}<r<R_{{\rm{max}},k}, \\
            e^{-\pi\tilde{\lambda}_j\left[\left(\hat{P}_j\hat{G}_j\hat{B}_j\hat{E}_j\right)^{\frac{2}{\alpha_j}}r^{\frac{2}{\hat{\alpha}_j}}-R_{{\rm{min}},j}^2\right]}, &L_{k,j}<r<U_{k,j},\\
            1, &R_{{\rm{min}},k}<r<L_{k,j},
        \end{cases}
    \end{align}
    \begin{equation}
        U_{k,j}=\left(\frac{R_{{\rm{max}},j}^{\alpha_j}}{\hat{P}_j\hat{G}_j\hat{B}_j\hat{E}_j}\right)^{\frac{1}{\alpha_k}},\; and \; L_{k,j}=\left(\frac{R_{{\rm{min}},j}^{\alpha_j}}{\hat{P}_j\hat{G}_j\hat{B}_j\hat{E}_j}\right)^{\frac{1}{\alpha_k}}
    \end{equation}
    with 
    \begin{equation}
        \hat{P}_j=\frac{P_j}{P_k},\;  \hat{G}_{j}=\frac{G_{j,1}}{G_{k,1}},\;\hat{B}_j=\frac{B_j}{B_k},\;\hat{E}_j=\frac{E_j}{E_k},\;\hat{\alpha}_j=\frac{\alpha_j}{\alpha_k}.
    \end{equation} \label{cell association}
\end{lemma}

\begin{proof}
    See Appendix \ref{proof of cell association}.
\end{proof}

\begin{lemma}
    When a typical user is associated with the $k$th tier, the PDF of $r$ conditioned on that at least one BS is visible is
    \begin{align}
        f_k(r|{\tilde \Phi}_k(\tilde{\mathcal{A}}_k)&>0,J=k)=\frac{1}{\mathbb{P}[J=k|{\tilde \Phi}_k(\tilde{\mathcal{A}}_k)>0]} \nonumber \\
        &\times\frac{2\pi\tilde{\lambda}_k}{1-e^{-\pi\tilde{\lambda}_k\left(R_{{\rm{max}},k}^2-R_{{\rm{min}},k}^2\right)}}r\prod_{j=1}^{K}{\sf P}_{k,j}.
    \end{align} \label{nearest distribution of kth tier}
\end{lemma}

\begin{proof}
    See Appendix \ref{proof of neareast distribtuion of kth tier}.
\end{proof}

These lemmas show how the network connection probability changes in terms of the critical system parameters including the network densities ${\tilde \lambda}_j$, the normalized bias factors ${\hat B}_j$, the normalized path-loss exponents ${\hat \alpha}_j$, the normalized transmit power ${\hat P}_j$, and the normalized average fading power ${\hat E}_j$.

\vspace{-0.2cm}

\subsection{Interference Laplace}
The analysis of the coverage probability of ISTNs relies significantly on understanding the distribution of the aggregated interference, specifically the intra-tier interference power. The subsequent lemma elucidates the Laplace transform of the aggregated interference power, providing a crucial insight into the overall performance of ISTNs.

\begin{lemma}
    Let $I_k$ be the aggregated interference power of the $k$th tier as
    \begin{equation}
        I_k=\sum_{i=2}^{N_k}\tilde{G}_k H_{k,i}\lVert{\bf \tilde s}_{k,i}-\tilde{\mathbf{u}}_1\rVert^{-\alpha_k}.
    \end{equation}
    Then, the Laplace transform of the aggregated interference power plus normalized noise variance $I_k+\hat{\sigma}_k^2$ of the $k$th tier conditioned on that at least one BS is on annulus $\tilde{\mathcal{A}}_k$ is obtained as
    \begin{align}
        &\mathcal{L}_{I_k+\hat{\sigma}_k^2|{\tilde \Phi}_k(\tilde{\mathcal{A}}_k)>0}(s)=\exp\left(-s\hat{\sigma}_k^2-2\pi\tilde{\lambda}_k\int_{r}^{R_{{\rm{max}},k}}\Bigg[1\right. \nonumber \\
        &\left.\left.-\frac{(2b_k m_k)^{m_k}(1+2b_k s\tilde{G}_k v^{-\alpha_k})^{m_k-1}}{\left[(2b_k m_k+\Omega_k)(1+2b_k s\tilde{G}_k v^{-\alpha_k})-\Omega_k\right]^{m_k}}\right]v{\rm{d}}v\right).
    \end{align} \label{Laplace of interference power plus noise}
\end{lemma}

\begin{proof}
    See Appendix \ref{proof of laplace of interference power}.
\end{proof}

The aggregated interference power is significantly influenced by the fading distribution of the interfering links and the density ${\tilde \lambda}_k$. For instance, as the network density decreases, the interference Laplace increases, which helps to improve the coverage probability. However, the proximity distribution of the BS worsens as the network density scales down, posing a detrimental impact on the coverage probability.

\section{Coverage Probability} \label{sec:coverage probability analysis}
In this section, we first present the exact expression for the coverage probability of the $K$-tier ISTNs. Unfortunately, the exact expression is challenging for analytical computation due to its involvement with intricate integrals. To address this complexity, we introduce an approximate expression for the coverage probability, specifically designed to be computed with significantly less intricacy, particularly under predefined conditions.

\vspace{-0.2cm}

\subsection{Exact Expression}
The following theorem is the main result of this paper. 
\begin{theorem}
    The downlink coverage probability of the typical user is
    \begin{align}
        {\sf  P}^{\sf cov}(T)&=\sum_{k=1}^{K}{\sf P}_{k}^{\sf cov}(T|{\tilde \Phi}_k(\tilde{\mathcal{A}}_k)>0,J=k) \nonumber \\
        &\qquad\times\mathbb{P}[J=k|{\tilde \Phi}_k(\tilde{\mathcal{A}}_k)>0]\mathbb{P}[{\tilde \Phi}_k(\tilde{\mathcal{A}}_k)>0],
    \end{align}
    where $\mathbb{P}[{\tilde \Phi}_k(\tilde{\mathcal{A}}_k)>0]$ and $\mathbb{P}[J=k|{\tilde \Phi}_k(\tilde{\mathcal{A}}_k)>0]$ are given by Lemmas \ref{lem2} and \ref{cell association}. Furthermore, the coverage probability associated to the $k$th tier is given by
    \begin{align}
        &\!\!\!\!\!\!{\sf P}_{k}^{\sf cov}(T|{\tilde \Phi}_k(\tilde{\mathcal{A}}_k)>0,J=k) \nonumber \\
        &\!\!\!\!\!\!=\int_{R_{{{\rm{min}}},k}}^{R_{{{\rm{max}}},k}}\Bigg[1-\sum_{l=0}^{m_k-1}\frac{\zeta_k(l)\Gamma(l+1)}{(\beta_k-\delta_k)^{l+1}}\Bigg(1-\sum_{q=0}^{l}\frac{\nu_k^q}{q!}(-1)^q\;\times \nonumber \\
        &\!\!\!\!\!\frac{{\rm{d}}^q\mathcal{L}_{I_k+\hat{\sigma}_k^2|{\tilde \Phi}_k(\tilde{\mathcal{A}}_k)>0}(s)}{{\rm{d}}s^q}\Bigg|_{s=\nu_k}\Bigg)\Bigg]f_k(r|{\tilde \Phi}_k(\tilde{\mathcal{A}}_k)\!>\!0,J=k){\rm{d}}r, \label{eq:Exact form of conditional coverage probability}
    \end{align}
    where $\zeta_k(l)=\frac{Z_k(-\delta_k)^l(1-m_k)_l}{(l!)^2}$, $\nu_k=(\beta_k-\delta_k)r^{\alpha_k}T$, and $\Gamma(\cdot)$ is the gamma function. \label{theo:exact form of coverage}
\end{theorem}

 \begin{proof}
     See Appendix \ref{proof of Theorem 2}.
 \end{proof}

Our first key observation is the difference from \cite{andrews2011tractable} in the interference-limited regime. In \cite{andrews2011tractable}, the coverage probability does not depend on the density of BSs when noise is ignored. Even in the interference-limited regime, the coverage probability varies depending on the density of BSs. This captures a realistic environment by modeling the terrestrial network on the sphere of the Earth. Intuitively, since the sum of the probabilities of being associated with the $k$th tier of Lemma \ref{cell association} is 1 for all tiers, the coverage probability of Theorem \ref{theo:exact form of coverage} may be considered an internally dividing point for every single tier at first glance. However, due to the difference in conditional nearest BS distance distribution between Lemmas \ref{conditional nearest distance distribution} and \ref{nearest distribution of kth tier} for single tier and multi-tier, Theorem \ref{theo:exact form of coverage} properly reflects the effect of macro-diversity that occurs by cooperating multiple networks. This is because when an outage occurs in one network, there is still a possibility that communication can succeed in the remaining networks. This possibility is reflected in terms of different tiers in Lemma \ref{nearest distribution of kth tier}. Therefore, as the number of tiers increases, the effect of macro-diversity increases because they do not share a spectrum. This macro-diversity effect is verified through simulation results in Section \ref{sec:simulation}.

\vspace{-0.2cm}

\subsection{Approximated Expression}
The exact expression of Theorem \ref{theo:exact form of coverage} is computationally expensive when $m_k$ is large due to the derivative of $\mathcal{L}_{I_k+\hat{\sigma}_k^2|{\tilde \Phi}_k(\tilde{\mathcal{A}}_k)>0}(s)$. Therefore, we also provide an approximated form of the coverage probability.
 
\begin{theorem}
    The approximated coverage probability of the typical user conditioned on that a typical user is associated with the $k$th tier and at least one BS is visible in the $k$th tier is obtained as
    \begin{align}
        &{\sf P}_{k}^{\sf cov}(T|{\tilde \Phi}_k(\tilde{\mathcal{A}}_k)>0,J=k) \nonumber \\
        &\approx\int_{R_{{{\rm{min}}},k}}^{R_{{{\rm{max}}},k}}\left[1-\sum_{l=0}^{m_k-1}\frac{\zeta_k(l)\Gamma(l+1)}{(\beta_k-\delta_k)^{l+1}}\sum_{q=0}^{l+1}{l+1 \choose q}(-1)^q\right. \nonumber \\
        &\;\times\mathcal{L}_{I_k+\hat{\sigma}_k^2|{\tilde \Phi}_k(\tilde{\mathcal{A}}_k)>0}(\rho_k q)\Bigg]f_k(r|{\tilde \Phi}_k(\tilde{\mathcal{A}}_k)>0,J=k){\rm{d}}r,
    \end{align}
    where $\rho_k=\kappa(\beta_k-\delta_k)r^{\alpha_k}T$, and $\kappa$ is the tuning parameter. This tuning parameter $\kappa$ \label{Approximation of coverage probability} is adjustable only in the range of
    \begin{equation}
        \left(\Gamma(l+2)\right)^{-\frac{1}{l+1}}\le\kappa\le 1.
    \end{equation} \label{theo:approximation of coverage}
\end{theorem}

\begin{proof}
    See Appendix \ref{Proof of Theorem 3}.
\end{proof}

This approximated form can be calculated faster than the exact form as $m_k$ becomes larger. Therefore, the approximated expression can be useful in environment where shadowing is not heavy. However, the approximation is very close to the exact form, and we verify how close it is through simulation results in Section \ref{sec:simulation}.
\vspace{-0.2cm}

\subsection{Special Case}
To provide more clear intuition on the derived coverage probability expressions, we establish the approximated expression of the coverage probability in some special cases, which is stated in the following corollary.  

 \begin{corollary}
    When $K=2$ with $K_{\sf T}=1$ and $K_{\sf S}=1$, $\alpha_k=2$, and $H_{k,i}$ are distributed as Rayleigh, the coverage probability in the interference-limited regime is approximated as
    \begin{align}
        &{\sf  P}^{\sf cov}(T)\approx \left[1-e^{-\pi\tilde{\lambda}_1\left(R_{{\rm{max}},1}^2-R_{{\rm{min}},1}^2\right)}\right] \times \nonumber \\
        &\left[\frac{\Xi_1 e^{\chi_1}}{\psi_1+\omega_1}e^{\frac{(\psi_1+\omega_1)\left(L_{1,2}^2+R_{{\rm{min}},1}^2\right)}{2}}\sinh\left(\frac{(\psi_1+\omega_1)\left(L_{1,2}^2-R_{{\rm{min}},1}^2\right)}{2}\right)\right.\nonumber\\
        &\left.\quad+\frac{\Xi_1 e^{\chi}}{\psi_1+\omega}e^{\frac{(\psi_1+\omega)\left(U_{1,2}^2+L_{1,2}^2\right)}{2}}\sinh\left(\frac{(\psi_1+\omega)\left(U_{1,2}^2-L_{1,2}^2\right)}{2}\right)\right]\nonumber\\
        &\qquad\quad+\left[1-e^{-\pi\tilde{\lambda}_2\left(R_{{\rm{max}},2}^2-R_{{\rm{min}},2}^2\right)}\right]\times\nonumber\\
        &\left[\frac{\Xi_2 e^{\chi}}{\psi_2+\mu}e^{\frac{(\psi_2+\mu)\left(R_{{\rm{max}},2}^2+R_{{\rm{min}},2}^2\right)}{2}}\sinh\left(\frac{(\psi_2+\mu)\left(R_{{\rm{max}},2}^2-R_{{\rm{min}},2}^2\right)}{2}\right)\right],
    \end{align}
    where $\chi_k=\pi\tilde{\lambda}_k R_{{\rm{min}},k}^2$ and $\chi=\chi_1+\chi_2$. In addition, $\omega_j=-\pi\tilde{\lambda}_j\frac{P_j}{P_1}\frac{G_{j,1}}{G_{1,1}}\frac{B_j}{B_1}$, $\mu_j=-\pi\tilde{\lambda}_j\frac{P_j}{P_2}\frac{G_{j,1}}{G_{2,1}}\frac{B_j}{B_2}$, $\omega=\omega_1+\omega_2$, and $\mu=\mu_1+\mu_2$. Furthermore,  
        \begin{equation}
        \Xi_k=\frac{2\pi\tilde{\lambda}_k}{1-e^{-\pi\tilde{\lambda}_k\left(R_{{\rm{max}},k}^2-R_{{\rm{min}},k}^2\right)}}.
    \end{equation} \label{Special case} 
    and
    \begin{equation}
        \psi_k=-\pi\tilde{\lambda}_k T\tilde{G}_k\ln\left(\frac{T\tilde{G}_k+\left(\frac{R_{{\rm{max}},k}}{R_{{\rm{min}},k}+\varepsilon_k}\right)^2}{1+T\tilde{G}_k}\right),
    \end{equation}
   where $\varepsilon_k$ is an arbitrary correction constant for approximation.

\end{corollary}

\begin{proof}
The proof is analogous the proof in Appendix E of \cite{park2022tractable}. The key difference with \cite{park2022tractable} is the approximation method of the interference Laplace term in Lemma \ref{Laplace of interference power plus noise} using $R_{{\rm{min}},k}+\varepsilon_k$ instead of $R_{{\rm{min}},k}$. By tuning parameter $\epsilon_k$, we can make a tight approximation.
\end{proof}

Corollary 1 provides valuable insights into how the coverage probability varies with critical network parameters. Similar to the prior result in \cite{park2022tractable}, which is an alliterative analytical expression of Corollary 1, we note that in each network tier, the optimal density of satellite nodes for maximizing coverage probability diminishes with increasing network altitude. This result suggests that terrestrial networks necessitate a higher density of BSs than satellite networks to enhance coverage probability effectively.

Furthermore, Corollary 1 elucidates the impact of the biasing factor on network density. Essentially, increasing the biasing factor from 1 to $B_k>1$ can be conceptualized as adjusting the network density from  ${\tilde \lambda}_k$ to ${\tilde \lambda}_k B_k^{\frac{2}{\alpha}}$. Consequently, elevating the biasing factor in each tier effectively amplifies network density, thereby augmenting the likelihood of network connection for typical users. This observation resonates with our intuitive understanding of the biasing factor. The simulation section will provide a more detailed exploration of the biasing effect.


\section{Simulation Results} \label{sec:simulation}
In this section, we validate our derived expressions for the Theorem through a comparison with Monte Carlo simulations. The shadowed-Rician fading, employed as the fading channel model in our study, offers versatility by reflecting various fading channels based on parameter configurations. Consequently, we employ Rayleigh fading for the terrestrial network and the LMS channel parameters specified in Table \ref{tab:simulation parameters} for the satellite network. Given our assumption that $m$ is an integer, we adjust the actual LMS channel parameters to the nearest integer value for $m$. Unless explicitly mentioned otherwise, we utilize the parameters listed in Table \ref{tab:simulation parameters} for our simulations.

For the parameters pertaining to the satellite network, we draw upon values extracted from the link budget filed with the FCC \cite{FCC2023direct}. This filing was submitted by SpaceX for the operation of a direct-to-cellular system. Notably, we assume omnidirectional reception for users' antennas with a gain set to 0 dBi, denoted as $G_U=\bar{G}_U=1$. Our simulation is structured in two tiers to elucidate discernible trends, although it is inherently adaptable to multiple tiers without loss of generality.

\begin{table}[t!]
    \centering
    \caption{Simulation Parameters}
    \label{tab:simulation parameters}
    \begin{tabular}{c|c|c}
    \hline
    Parameter & Terrestrial & Satellite \\
    \hline
        Carrier frequency $f_k$ [GHz] & 3.5 & 1.9925\\
        \hline
        Radius of the Earth $R_{\sf E}$ [km] & \multicolumn{2}{c}{6371}\\
        \hline
        Path-loss exponent $\alpha_k$ & 4 & 2 \\
        \hline
        Speed of light $c$ [km/s] & \multicolumn{2}{c}{$3\times 10^5$}\\
        \hline
        Noise spectral density $N_0$ [dBm/Hz] & \multicolumn{2}{c}{-174}\\
        \hline
        Tx antenna gain for desired link $G_k^t$ [dBi] & 0 & 38\\
        \hline
        Tx antenna gain for interference $\bar{G}_k^t$ [dBi] & 0 & 28\\
        \hline
        Tx power $P_k$ [dBm] & 46 & 50\\
        \hline
        Bandwidth $W_k$ [MHz] & 100 & 5\\
        \hline
        Altitude $h_k$ [km] & 0.03 & 530\\
        \hline
        {\color{black} Visible elevation angle} $\theta_{k}$ [rad] & \multicolumn{2}{c}{$\theta_{k}^{\rm{max}}$} \\
        \hline
    \end{tabular}
    
    \medskip
    
    \begin{tabular}{c|c|c|c}
    \hline
    \multicolumn{4}{c}{LMS Channel Parameters \cite{abdi2003new}} \\
        \hline
        Shadowing & $m$ & $b$ & $\Omega$ \\
        \hline
        Frequent Heavy Shadowing (FHS) & 1 & 0.063 & $8.97\times 10^{-4}$\\
        \hline
        Average Shadowing (AS) & 10 & 0.126 & 0.835\\
        \hline
        Infrequent light Shadowing (ILS) & 19 & 0.158 & 1.29\\
        \hline
    \end{tabular} \vspace{-0.23cm}
\end{table}

\begin{figure*}[t!]
  \centering
  \includegraphics[width=0.9\textwidth]{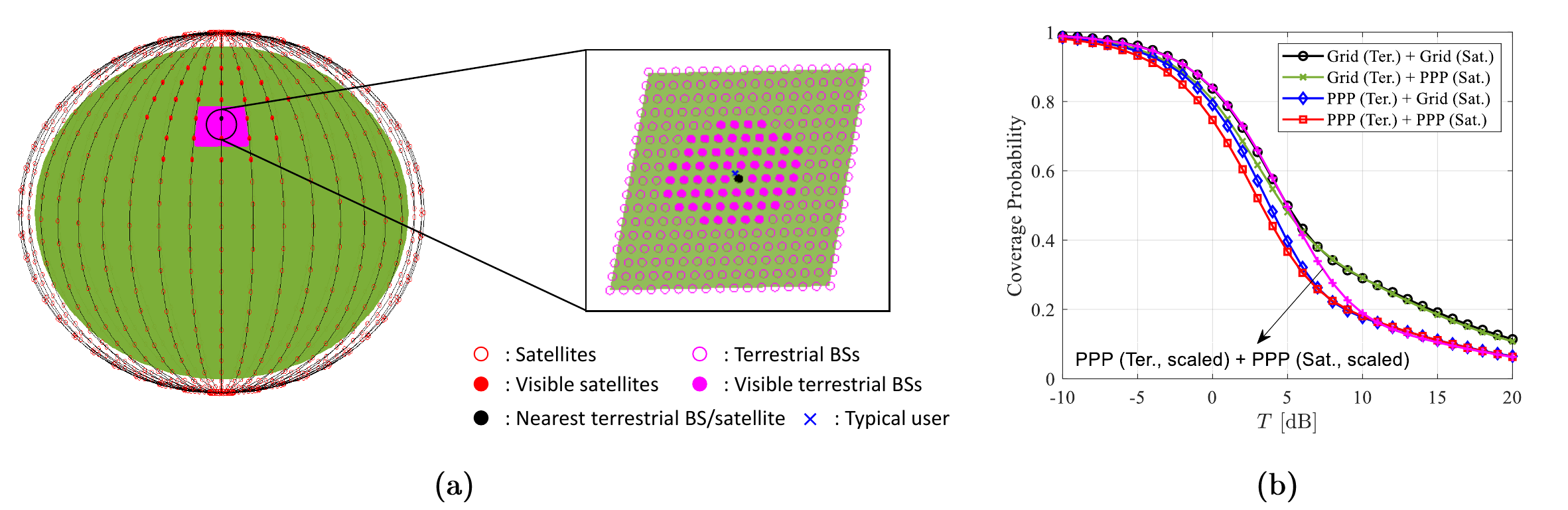}
  \caption{(a) Example of Walker star constellation grid model. The right figure shows a magnified terrestrial network. The green sphere, red unfilled circle, red filled circle, magenta unfilled circle, magenta filled circle, black filled circle, and blue marker represent the Earth, satellites, visible satellites, terrestrial BSs, visible terrestrial BSs, nearest satellite/terrestrial BS, and the typical user, respectively. (b) Coverage probability comparison between the proposed PPP model and the Walker star constellation grid model. The number of satellites is set to 1,000, and the number of orbits is set to 20. The locations of the satellites are placed at 36 degrees north latitude and 126 degrees east longitude. Table \ref{tab:simulation parameters} considers the average shadowing.}
  \label{fig:comparison between grid and PPP}
\end{figure*}

\vspace{-0.2cm}

\subsection{Network Model Validation}
To validate the proposed ISTN model, we compare our PPP-based ISTN model with the classical grid network model, as shown in Fig. \ref{fig:comparison between grid and PPP}-(a). We consider the Walker star constellation for terrestrial and satellite networks in a grid model. Therefore, our grid model differs from the grid model in \cite{andrews2011tractable} in that the height of the terrestrial BS is considered. In the grid model, we assumed that a typical user is randomly located in the area where the reference satellite is the nearest. Therefore, once the user's location is determined, the user is served by the reference satellite or the nearest terrestrial BS. Because the number of visible satellites and terrestrial BSs varies depending on the location of a typical user, we set the density in the PPP-based model to the average value of several snapshots. 

Fig. \ref{fig:comparison between grid and PPP}-(b) shows simulation results for four possible combinations of grid- and PPP-based satellite and terrestial network models. As mentioned in \cite{andrews2011tractable}, the distance between the nearest BS and interference BSs may be close in the PPP-based model, while the grid model ensures that a certain distance separates interference BSs. Therefore, unless the grid spacing is very close, the PPP-based model can serve as the lower bound on the coverage probability for the grid model, as shown in Fig. \ref{fig:comparison between grid and PPP}-(b). Both models have similar coverage tendencies, which validates our PPP-based network model. Additionally, we can similarly match the coverage probability in low SINR threshold regions by adjusting the density of the network. We also emphasize that the PPP-based satellite network model has shown to be accurate for analyzing the commercial LEO satellite network such as Starlink \cite{park2022tractable}. 
\vspace{-0.2cm}

\subsection{Verification and Effect of Shadowing}
\begin{figure}[t!]
  \centering
  \includegraphics[width=0.45\textwidth]{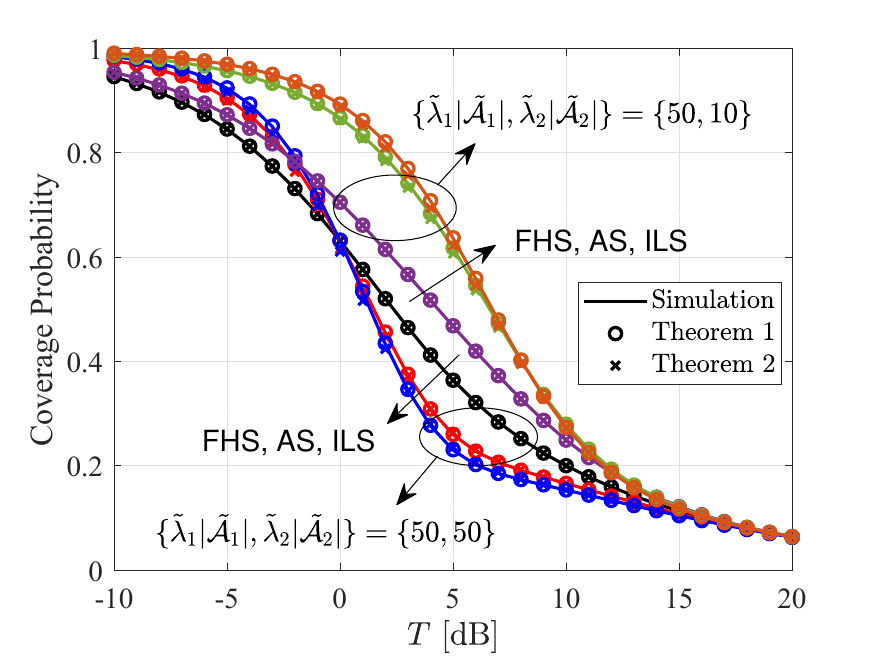}
  \caption{Coverage probability in a two-tier ISTN for verifying Theorem \ref{theo:exact form of coverage} and \ref{theo:approximation of coverage} ($B_1=B_2=1$).} \vspace{-0.28cm}
  \label{fig:verification of Theorem}
\end{figure}

Fig. \ref{fig:verification of Theorem} shows the simulation results to verify Theorem \ref{theo:exact form of coverage} and \ref{theo:approximation of coverage} according to shadowing and satellite density. The analytical expressions derived in Theorem \ref{theo:exact form of coverage} and \ref{theo:approximation of coverage} perfectly match and are close to the simulation results. When $m_2=1$, that is, FHS is considered, the approximated coverage probability matches well with the exact expression of the coverage probability. Therefore, Theorem \ref{theo:exact form of coverage} and \ref{theo:approximation of coverage} are identical for FHS. However, if $m_2\ge 2$, approximation error occurs no matter how well $\kappa$ is adjusted. For AS and ILS, although the approximation errors occur, they are negligible and allow us to obtain coverage probabilities faster than exact form with lower computational complexity. 

When $\tilde{\lambda}_2|\tilde{\mathcal{A}}_2|=50$, in the low SINR threshold region, coverage increases as shadowing becomes lighter, that is, as the LMS channel parameters $m_2$, $b_2$, and $\Omega_2$ increases. This is because the influence of the nearest BS is greater than the interference in a low SINR threshold region. Therefore, a light shadowing environment with a strong LoS path of the nearest BS shows better performance. However, this tendency acts in the opposite direction above a certain SINR threshold, which is due to interference. Interference from satellites other than the nearest BS in a light shadowing environment reduces coverage probability. On the other hand, when $\tilde{\lambda}_2|\tilde{\mathcal{A}}_2|=10$, in all SINR threshold regions, the coverage probability is highest in ILS and lowest in FHS because the number of satellites acting as interference is small from the first.

\begin{figure}[t!]
  \centering
  \includegraphics[width=0.45\textwidth]{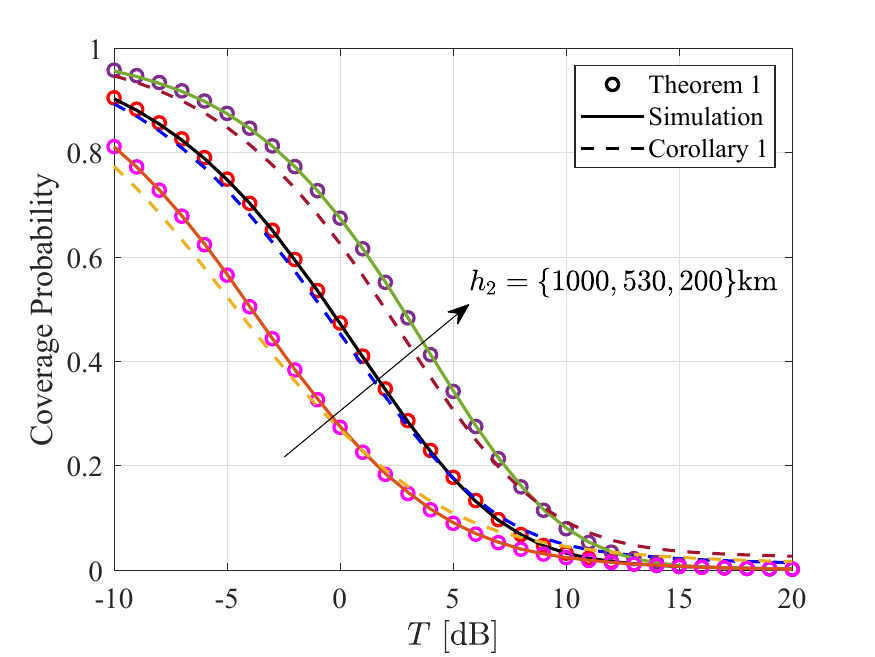}
  \caption{ Coverage probability for verifying Corollary \ref{Special case} ($B_1=B_2=1$, $\tilde{\lambda}_1|\tilde{\mathcal{A}}_1|=\tilde{\lambda}_2|\tilde{\mathcal{A}}_2|=30$).} \vspace{-0.3cm}
  \label{fig:verification of Corollary}
\end{figure}

Fig. \ref{fig:verification of Corollary} shows the simulation results to verify Corollary \ref{Special case} according to the altitude of the satellite. We set the simulation parameters as  $\{\varepsilon_1,\varepsilon_2\}=\{2.9282,1.4089\}$ for $h_2=200$ km, $\{\varepsilon_1,\varepsilon_2\}=\{1.9521,0\}$ for $h_2=530$ km, and $\{\varepsilon_1,\varepsilon_2\}=\{2.1474,1.3535\}$ for $h_2=1000$ km. Compared to the exact form, the coverage probability derived in Corollary \ref{Special case} exhibits similar trends, which shows the tightness of the derived closed form expression in some special cases. As can be seen, we show that lowering satellite altitude provides a higher coverage performance gain for ISTNs. 

\vspace{-0.2cm}

\subsection{Effects of Biasing and Density}
As can be seen in Fig. \ref{fig:verification of Theorem}, the coverage trend depends not only on the channel parameters but also on the density. Additionally, the main parameter that determines coverage performance is the biasing factor. Therefore, we provide simulation results to examine the influence of these two parameters on coverage performance. In this subsection, AS is considered for shadowing. 
\begin{figure}[t!]
  \centering
  \includegraphics[width=0.45\textwidth]{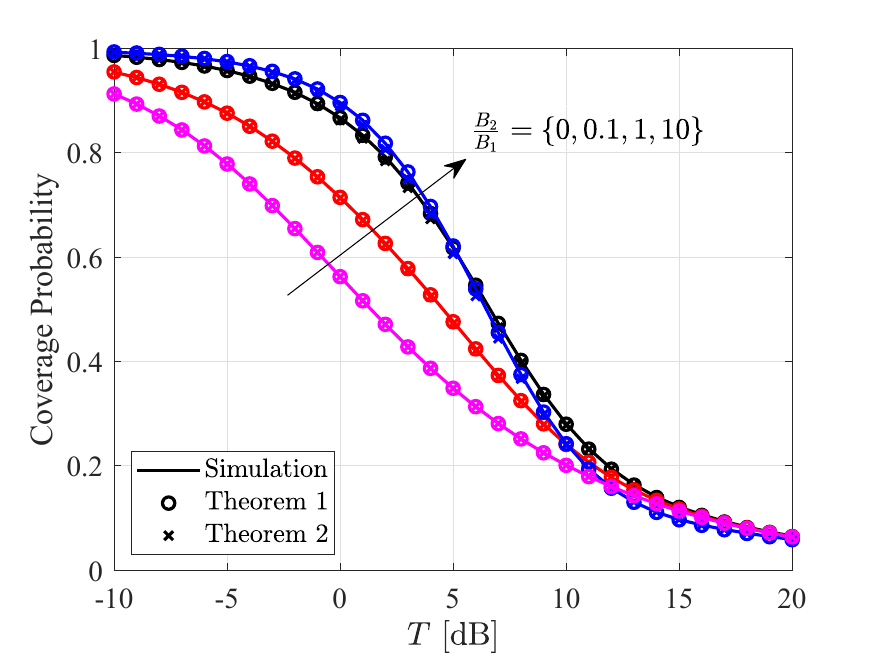}
  \caption{Coverage probability for varying the ratio of the biasing factor in a two-tier ISTN (AS, $\{\tilde{\lambda}_1|\tilde{\mathcal{A}}_1|,\tilde{\lambda}_2|\tilde{\mathcal{A}}_2|\}=\{50,10\}$).} \vspace{-0.2cm}
  \label{fig:Coverage_Probability_biasing}
\end{figure}
\begin{figure}[t!]
  \centering
  \includegraphics[width=0.45\textwidth]{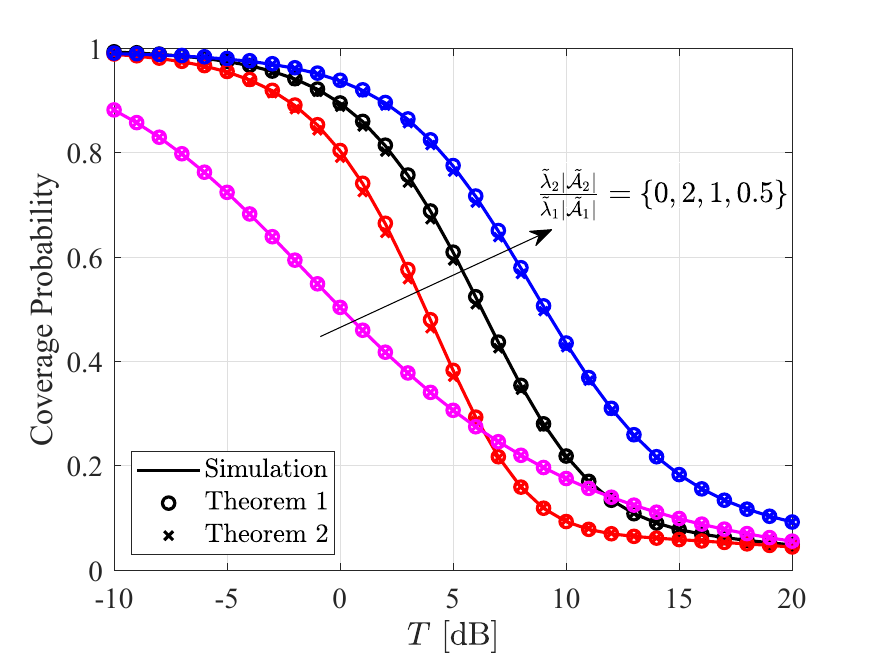}
  \caption{Coverage probability for varying the ratio of the density in a two-tier ISTN (AS, $\tilde{\lambda}_1|\tilde{\mathcal{A}}_1|=10$).}
  \label{fig:Coverage_Probability_density}  \vspace{-0.3cm}
\end{figure}

The simulation results depicted in Fig. \ref{fig:Coverage_Probability_biasing} illustrate the impact of varying biasing factor ratios. A ratio of 0 signifies exclusive consideration of the terrestrial network, while higher ratios denote increased user association with the satellite network. In the terrestrial network, Rayleigh fading ($m_1=1$) is taken into account. When the biasing factor ratio is 0, aligning with the terrestrial network alone, Theorems \ref{theo:exact form of coverage} and \ref{theo:approximation of coverage} yield identical results. At a ratio of 0.1, predominantly terrestrial network associations prevail due to its higher biasing factor, leading to a close match between the theorems. However, as the biasing factor tilts towards the satellite network, there's a threshold beyond which Theorem \ref{theo:approximation of coverage} diverges slightly from Theorem \ref{theo:exact form of coverage}, though still exhibiting notable consistency. Particularly, when $\tilde{\lambda}_1|\tilde{\mathcal{A}}_1|=50$ and $\tilde{\lambda}_2|\tilde{\mathcal{A}}_2|=10$, indicating a relatively lower satellite density compared to terrestrial BS, prioritizing collaboration with less interference-prone satellite networks is advisable.


Next, Fig. \ref{fig:Coverage_Probability_density} illustrates the simulation results based on the density ratio. Similar to Fig. \ref{fig:Coverage_Probability_biasing}, a ratio of 0 indicates consideration solely of the terrestrial network. When the density ratio is 0.5, the coverage performance surpasses that of the terrestrial network alone across all SINR thresholds $T$. Conversely, at density ratios of 1 and 2, beyond a certain threshold, the collaboration of both networks exhibits diminished coverage probability due to increasing interference from additional satellites. However, given the generally higher coverage probability across most threshold ranges, integrating both networks would be advantageous for most systems. As observed in Fig. \ref{fig:Coverage_Probability_biasing} and  Fig. \ref{fig:Coverage_Probability_density}, coverage varies depending on the biasing factor and density, suggesting an optimal ratio for these parameters to maximize coverage under different circumstances.

\begin{figure}[t!]
  \centering
  \includegraphics[width=0.45\textwidth]{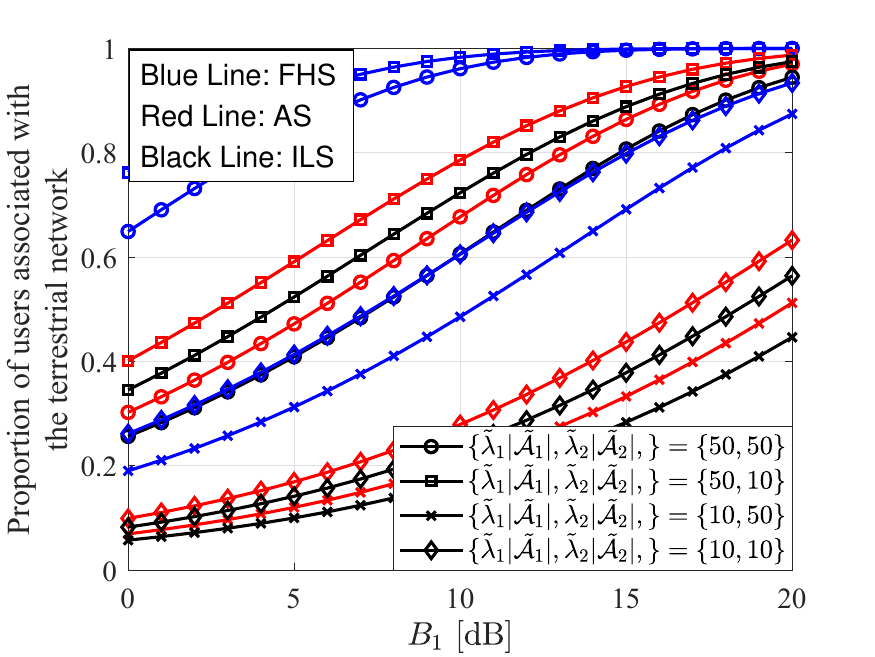}
  \caption{Proportion of users associated with the terrestrial network for varying biasing factor of the terrestrial network and density in a two-tier ISTN ($B_2=1, T=0$ dB).} \vspace{-0.3cm}
  \label{fig:Users_proportion_biasing}
\end{figure}
\begin{figure}[t!]
  \centering
  \includegraphics[width=0.45\textwidth]{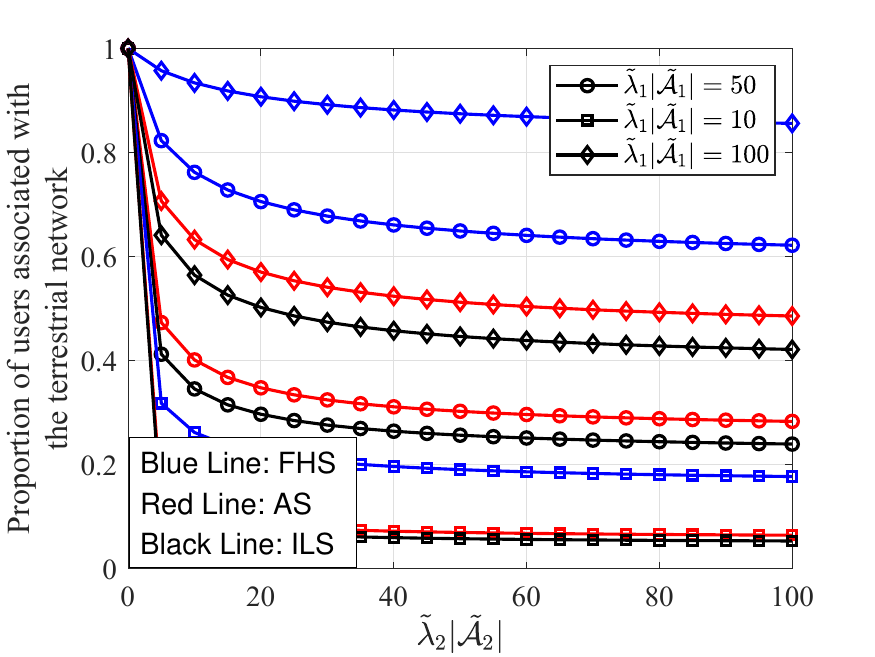}
  \caption{Proportion of users associated with the terrestrial network for varying density in a two-tier ISTN ($B_1=B_2=1, T=0$ dB).} \vspace{-0.2cm}
  \label{fig:Users_proportion_density} 
\end{figure}

Fig. \ref{fig:Users_proportion_biasing} and Fig. \ref{fig:Users_proportion_density} depict the proportion of users linked to the terrestrial network based on key parameters such as shadowing, biasing factor, and density. Specifically, Fig. \ref{fig:Users_proportion_biasing} illustrates how the proportion of users associated with the terrestrial network varies with the biasing factor of the terrestrial network across different shadowing and density scenarios. A higher value of $B_1$ implies a greater bias towards the terrestrial network, consequently leading to an increase in the proportion of users associated with it. Essentially, lighter shadowing corresponds to a stronger LOS path power. Therefore, under similar density conditions, the proportion of users associated with the terrestrial network is highest in scenarios with FHS and lowest in scenarios with Infrequent Light Shadowing ILS. This trend arises because a higher biasing factor amplifies the influence of the terrestrial network, making it more attractive to users. Moreover, in environments with high network density, a significant portion of users tends to be associated with networks exhibiting greater density. Notably, the density of terrestrial BSs significantly impacts the distribution of users among networks, highlighting the sensitivity of terrestrial BS density to the proportion of users associated with each network.

Fig. \ref{fig:Users_proportion_density} illustrates how the proportion of users linked with the terrestrial network varies concerning the satellite density across different terrestrial base station densities. Comparably, Fig. \ref{fig:Users_proportion_biasing} echoes this trend, revealing that the highest proportion of users affiliated with the terrestrial network occurs in FHS scenarios, while the lowest is observed in ILS scenarios, all under identical terrestrial base station densities. Evidently, with an escalation in satellite density, there is a corresponding increase in the proportion of users associated with the satellite network. However, as depicted in the case of $\tilde{\lambda}_1|\tilde{\mathcal{A}}_1|=10$, this proportion reaches a saturation point beyond a certain satellite density. Moreover, the density of terrestrial base stations also significantly influences the proportion of users linked to the terrestrial network. Consequently, both the terrestrial and satellite network densities exert substantial impacts on coverage metrics.

\section{Conclusion} \label{sec:conclusion}

In this paper, we have introduced a analytical method for assessing downlink coverage performance of the $K$-tier ISTNs. Our approach involves modeling the spatial distribution of BSs of each network tier  using homogeneous PPPs on concentric spheres. This modeling accommodates multiple tiers within both terrestrial and satellite networks while maintaining identical network geometry shapes. Under this unified network model, we have derived the expressions for the downlink coverage probability of ISTNs, incorporating essential network design parameters such as network density, path-loss exponent, fading parameters, and network association bias factors. Through simulations, we validate the accuracy of these derived expressions and analyze the impacts of various network parameters. Our tractable approach provides valuable insights for system design and parameter optimization in multi-band ISTN deployment, facilitating informed decision-making and efficient network management.

{\appendices

\section{Proof of Lemma \ref{conditional nearest distance distribution}} \label{proof of conditional nearest distance distribution}
The conditional nearest BS distance distribution of the $k$th tier can be calculated from the complementary cumulative distribution function (CCDF) conditioned on that ${\tilde \Phi}_k(\tilde{\mathcal{A}}_k)>0$. Let $D_k$ be the distance between the typical user and the nearest BS of the $k$th tier. The conditional CCDF of the nearest BS distance of the $k$th tier can be obtained as
\begin{align}
    F_k^c(r|{\tilde \Phi}_k(\tilde{\mathcal{A}}_k)>0)    &=\mathbb{P}[D_k>r|{\tilde \Phi}_k(\tilde{\mathcal{A}}_k)>0] \nonumber \\
    &\overset{(a)}{=}\frac{\mathbb{P}[{\tilde \Phi}_k(\tilde{\mathcal{A}}_k^r)=0,{\tilde \Phi}_k(\tilde{\mathcal{A}}_k)>0]}{\mathbb{P}[{\tilde \Phi}_k(\tilde{\mathcal{A}}_k)>0]} \nonumber \\
    &\overset{(b)}{=}\frac{\mathbb{P}[{\tilde \Phi}_k(\tilde{\mathcal{A}}_k^r)=0]\mathbb{P}[{\tilde \Phi}(\tilde{\mathcal{A}}_k\backslash\tilde{\mathcal{A}}_k^r)>0]}{\mathbb{P}[{\tilde \Phi}_k(\tilde{\mathcal{A}}_k)>0]} \nonumber \\
    &=\frac{\mathbb{P}[{\tilde \Phi}_k(\tilde{\mathcal{A}}_k^r)=0]-\mathbb{P}[{\tilde \Phi}_k(\tilde{\mathcal{A}}_k)=0]}{1-\mathbb{P}[{\tilde \Phi}_k(\tilde{\mathcal{A}}_k)=0]} \nonumber \\
    &\overset{(c)}{=}\frac{e^{-\tilde{\lambda}_k\pi(r^2-R_{{{\rm{min}},k}}^2)}-e^{-\tilde{\lambda}_k\pi(R_{{{\rm{max}},k}}^2-R_{{{\rm{min}},k}}^2)}}{1-e^{-\tilde{\lambda}_k\pi(R_{{{\rm{max}},k}}^2-R_{{{\rm{min}},k}}^2)}},
\end{align}
where $\tilde{\mathcal{A}}_k^r$ is the partial spherical cap defined as
\begin{equation}
    \tilde{\mathcal{A}}_k^r=\left\{ {\bf \tilde s}_k\in\mathbb{R}^2:\lVert{\bf \tilde s}_k-\tilde{\mathbf{u}}_1\rVert\le r\right\}\subset\tilde{\mathcal{A}}_k,
\end{equation}
(a) is by Bayes' theorem, (b) follows from the independence between disjoint sets ${\tilde \Phi}_k(\tilde{\mathcal{A}}_k^r)$ and ${\tilde \Phi}_k(\tilde{\mathcal{A}_k}\backslash\tilde{\mathcal{A}}_k^r)$ and (c) is by the void probability of Lemma 1. The conditional nearest BS distance distribution of the $k$th tier can be calculated by taking the derivative with respect to $r$, which completes the proof. 

\section{Proof of Lemma \ref{cell association}} \label{proof of cell association}
In order for a typical user to be associated with the $k$th tier, the ERP of the $k$th tier must be the largest. Namely, $P^{\rm{eff}}_k>P^{\rm{eff}}_j$ for all $j\in\{1,...,K\}\backslash \{k\}$. Therefore, Lemma \ref{cell association} is calculated as
\begin{align}
    &\mathbb{P}[J=k|{\tilde \Phi}_k(\tilde{\mathcal{A}}_k)>0]\nonumber\\
    &=\mathbb{E}\left[\mathbb{P}\left[P_k^{\rm{eff}}>\max_{j,j\ne k}P_j^{\rm{eff}}\Big|{\tilde \Phi}_k(\tilde{\mathcal{A}}_k)>0\right]\right]\nonumber\\
    &=\mathbb{E}\left[\prod_{j=1,j\ne k}^{K}\mathbb{P}\left[P_k^{\rm{eff}}>P_j^{\rm{eff}}\Big|{\tilde \Phi}_k(\tilde{\mathcal{A}}_k)>0\right]\right]\nonumber\\
    &=\mathbb{E}\left[\prod_{j=1,j\ne k}^{K}\mathbb{P}\left[D_j>\left(\hat{P}_j\hat{G}_j\hat{B}_j\hat{E}_j\right)^{\frac{1}{\alpha_j}}D_k^{\frac{1}{\hat{\alpha}_j}}\Big|{\tilde \Phi}_k(\tilde{\mathcal{A}}_k)>0\right]\right]\nonumber\\
    &=\int_{R_{{\rm{min}},k}}^{R_{{\rm{max}},k}}\prod_{j=1,j\ne k}^{K}\mathbb{P}\bigg[D_j>\left(\hat{P}_j\hat{G}_j\hat{B}_j\hat{E}_j\right)^{\frac{1}{\alpha_j}}r^{\frac{1}{\hat{\alpha}_j}}\nonumber\\
    &\qquad\qquad\qquad \Big|{\tilde \Phi}_k(\tilde{\mathcal{A}}_k)>0\bigg]f_k(r|{\tilde \Phi}_k(\tilde{\mathcal{A}}_k)>0){\rm{d}}r.
\end{align}
In our model, because there is a range in which the BS can be located, the probability in the last equation varies depending on $r$. The interval of integration is divided into three cases according to the following conditions:
\begin{equation}
    \begin{cases}
    R_{{\rm{max}},j}<\left(\hat{P}_j\hat{G}_j\hat{B}_j\hat{E}_j\right)^{\frac{1}{\alpha_j}}r^{\frac{1}{\hat{\alpha}_j}},\\
    R_{{\rm{min}},j}\le \left(\hat{P}_j\hat{G}_j\hat{B}_j\hat{E}_j\right)^{\frac{1}{\alpha_j}}r^{\frac{1}{\hat{\alpha}_j}}\le R_{{\rm{max}},j},\\
    \left(\hat{P}_j\hat{G}_j\hat{B}_j\hat{E}_j\right)^{\frac{1}{\alpha_j}}r^{\frac{1}{\hat{\alpha}_j}}< R_{{\rm{min}},j}.
    \end{cases}
\end{equation}
Therefore,
\begin{align}
    &\mathbb{P}\left[D_j>\left(\hat{P}_j\hat{G}_j\hat{B}_j\hat{E}_j\right)^{\frac{1}{\alpha_j}}r^{\frac{1}{\hat{\alpha}_j}}\Big|{\tilde \Phi}_k(\tilde{\mathcal{A}}_k)>0\right]\nonumber\\
    &\!\!\!\!\!\!=\begin{cases}
            0, &U_{k,j}<r<R_{{\rm{max}},k}, \\
            e^{-\pi\tilde{\lambda}_j\left[\left(\hat{P}_j\hat{G}_j\hat{B}_j\hat{E}_j\right)^{\frac{2}{\alpha_j}}r^{\frac{2}{\hat{\alpha}_j}}-R_{{\rm{min}},j}^2\right]}, &L_{k,j}<r<U_{k,j},\\
            1, &R_{{\rm{min}},k}<r<L_{k,j},
    \end{cases}
\end{align}
where
\begin{equation}
        \!U_{k,j}=\left(\frac{R_{{\rm{max}},j}^{\alpha_j}}{\hat{P}_j\hat{G}_j\hat{B}_j\hat{E}_j}\right)^{\frac{1}{\alpha_k}},\;{\rm{and}} \; L_{k,j}=\left(\frac{R_{{\rm{min}},j}^{\alpha_j}}{\hat{P}_j\hat{G}_j\hat{B}_j\hat{E}_j}\right)^{\frac{1}{\alpha_k}}.
\end{equation}
When $j=k$, $U_{k,k}=R_{{\rm{max}},k}$, $L_{k,k}=R_{{\rm{min}},k}$, and
\begin{align}
    \mathbb{P}\bigg[D_k>\Big(\hat{P}_k\hat{G}_k\hat{B}_k\hat{E}_k\Big)^{\frac{1}{\alpha_k}}r^{\frac{1}{\hat{\alpha}_k}}\Big|{\tilde \Phi}_k(\tilde{\mathcal{A}}_k)>0\bigg]=e^{-\pi\tilde{\lambda}_k\left(r^{2}-R_{{\rm{min}},k}^2\right)}.
\end{align}
As a result,
\begin{align}
    &\mathbb{P}[J=k|{\tilde \Phi}_k(\tilde{\mathcal{A}}_k)>0]=\frac{2\pi\tilde{\lambda}_k}{1-e^{-\pi\tilde{\lambda}_k\left(R_{{\rm{max}},k}^2-R_{{\rm{min}},k}^2\right)}}\int_{R_{{\rm{min}},k}}^{R_{{\rm{max}},k}}\nonumber\\
    &r\prod_{j=1}^{K}\mathbb{P}\left[D_j>\left(\hat{P}_j\hat{G}_j\hat{B}_j\hat{E}_j\right)^{\frac{1}{\alpha_j}}r^{\frac{1}{\hat{\alpha}_j}}\Big|{\tilde \Phi}_k(\tilde{\mathcal{A}}_k)>0\right]{\rm{d}}r.
\end{align}

\section{Proof of Lemma \ref{nearest distribution of kth tier}} \label{proof of neareast distribtuion of kth tier}
We start by finding the conditional CCDF similar to the proof of Lemma \ref{conditional nearest distance distribution} in Appendix \ref{proof of conditional nearest distance distribution}. Therefore,
\begin{align}
    \!\!\!\!F_k^c(r|{\tilde \Phi}_k(\tilde{\mathcal{A}}_k)>0,J=k)&=\mathbb{P}[D_k>r|{\tilde \Phi}_k(\tilde{\mathcal{A}}_k)>0,J=k] \nonumber \\
    &=\frac{\mathbb{P}[D_k>r,J=k|{\tilde \Phi}_k(\tilde{\mathcal{A}}_k)>0]}{\mathbb{P}[J=k|{\tilde \Phi}_k(\tilde{\mathcal{A}}_k)>0]}.
\end{align}
The joint probability that $D_k>r$ and a typical user is associated with the $k$th tier is
\begin{align}
    &\mathbb{P}[D_k>r,J=k|{\tilde \Phi}_k(\tilde{\mathcal{A}}_k)>0]\nonumber\\
    &\!\!=\mathbb{P}[D_k>r,P_k^{\rm{eff}}>\max_{j,j\ne k}P_j^{\rm{eff}}|{\tilde \Phi}_k(\tilde{\mathcal{A}}_k)>0]\nonumber\\
    &\!\!=\int_{r}^{\infty}\!\!\prod_{j=1,j\ne k}^{K}\!\mathbb{P}\left[P_k^{\rm{eff}}\!>\!P_j^{\rm{eff}}\Big|{\tilde \Phi}_k(\tilde{\mathcal{A}}_k)>0\right]\!f_k(x|{\tilde \Phi}_k(\tilde{\mathcal{A}}_k)\!>\!0){\rm{d}}x\nonumber\\
    &\!\!=\int_{r}^{\infty}\prod_{j=1,j\ne k}^{K}\mathbb{P}\left[D_j>\left(\hat{P}_j\hat{G}_j\hat{B}_j\hat{E}_j\right)^{\frac{1}{\alpha_j}}x^{\frac{1}{\hat{\alpha}_j}}\Big|{\tilde \Phi}_k(\tilde{\mathcal{A}}_k)>0\right]\nonumber\\
    &\qquad\qquad\qquad\qquad\qquad\qquad\qquad\times f_k(x|{\tilde \Phi}_k(\tilde{\mathcal{A}}_k)>0){\rm{d}}x\nonumber\\
    &\!\!=\frac{2\pi\tilde{\lambda}_k}{1-e^{-\pi\tilde{\lambda}_k\left(R_{{\rm{max}},k}^2-R_{{\rm{min}},k}^2\right)}}\int_{r}^{\infty}x\nonumber\\
    &\;\;\times\prod_{j=1}^{K}\mathbb{P}\left[D_j>\left(\hat{P}_j\hat{G}_j\hat{B}_j\hat{E}_j\right)^{\frac{1}{\alpha_j}}x^{\frac{1}{\hat{\alpha}_j}}\Big|{\tilde \Phi}_k(\tilde{\mathcal{A}}_k)>0\right]{\rm{d}}x.
\end{align}
We can obtain the PDF of $r$ conditioned on that a typical user is associated with the $k$th tier and at least one BS is visible in the $k$th tier by differentiating the conditional CDF as
\begin{align}
    &f_k(r|{\tilde \Phi}_k(\tilde{\mathcal{A}}_k)>0,J=k) \nonumber \\ 
    &=\frac{1}{\mathbb{P}[J=k|{\tilde \Phi}_k(\tilde{\mathcal{A}}_k)>0]}\times\frac{2\pi\tilde{\lambda}_k}{1-e^{-\pi\tilde{\lambda}_k\left(R_{{\rm{max}},k}^2-R_{{\rm{min}},k}^2\right)}}\nonumber\\
    &\;\;\times r\prod_{j=1}^{K}\mathbb{P}\left[D_j>\left(\hat{P}_j\hat{G}_j\hat{B}_j\hat{E}_j\right)^{\frac{1}{\alpha_j}}r^{\frac{1}{\hat{\alpha}_j}}\Big|{\tilde \Phi}_k(\tilde{\mathcal{A}}_k)>0\right],
\end{align}
for $R_{{\rm{min}},k}<r<R_{{\rm{max}},k}$.

\section{Proof of Lemma \ref{Laplace of interference power plus noise}} \label{proof of laplace of interference power}
We first assume that the distance from the nearest BS of the $k$th tier to typical user is $r$. Therefore, the condition for this distance from typical user to nearest BS is also additionally required. As a result, the Laplace transform of the aggregated interference plus noise power conditioned on that the distance from a typical user to the nearest BS of the $k$th tier is $r$ and at least one BS exist in the annulus $\tilde{\mathcal{A}}_k$ is calculated as
\begin{align}
    &\mathcal{L}_{I_k+\hat{\sigma}_k^2|{\tilde \Phi}_k(\tilde{\mathcal{A}}_k)>0}(s) \nonumber \\
    &=\mathbb{E}[e^{-s(I_k+\hat{\sigma}_k^2)}|{\tilde \Phi}_k(\tilde{\mathcal{A}}_k)>0,\lVert{\bf \tilde s}_{k,1}-\tilde{\mathbf{u}}_1\rVert=r] \nonumber \\
    &=e^{-s\hat{\sigma}_k^2}\mathbb{E}[e^{-sI_k}|{\tilde \Phi}_k(\tilde{\mathcal{A}}_k)>0,\lVert{\bf \tilde s}_{k,1}-\tilde{\mathbf{u}}_1\rVert=r] \nonumber \\
    &\overset{(a)}{=}\exp\left(-s\hat{\sigma}_k^2-\tilde{\lambda}_k\int_{v\in\tilde{\mathcal{A}}_k/\tilde{\mathcal{A}}_k^r}\mathbb{E}\left[1-e^{-s\tilde{G}_k H_{k,i} v^{-\alpha_k}}\right]{\rm{d}}v\right) \nonumber \\
    &\overset{(b)}{=}\exp\Bigg(-s\hat{\sigma}_k^2-\tilde{\lambda}_k\int_{v\in\tilde{\mathcal{A}}_k/\tilde{\mathcal{A}}_k^r}1- \nonumber \\
    &\qquad\quad\;\frac{(2b_k m_k)^{m_k}(1+2b_k s\tilde{G}_k v^{-\alpha_k})^{m_k-1}}{\left[(2b_k m_k+\Omega_k)(1+2b_k s\tilde{G}_k v^{-\alpha_k})-\Omega_k\right]^{m_k}}{\rm{d}}v\Bigg) \nonumber \\
    &\overset{(c)}{=}\exp\Bigg(-s\hat{\sigma}_k^2-2\pi\tilde{\lambda}_k\int_{r}^{R_{{\rm{max}},k}}\Bigg[1- \nonumber \\
    &\quad\frac{(2b_k m_k)^{m_k}(1+2b_k s\tilde{G}_k v^{-\alpha_k})^{m_k-1}}{\left[(2b_k m_k+\Omega_k)(1+2b_k s\tilde{G}_k v^{-\alpha_k})-\Omega_k\right]^{m_k}}\Bigg]v{\rm{d}}v\Bigg).
\end{align}
where (a) follows from a Laplace functional of the PPP \cite{haenggi2009stochastic}, which is also known as the probability generating functional (PGFL) of the PPP \cite{westcott1972probability}, (b) is applied the moment generating function (MGF) of the shadowed-Rician fading \cite{abdi2003new}, and (c) is the conversion of the area corresponding to the surface area of the annulus into polar coordinates.

\section{Proof of Theorem \ref{theo:exact form of coverage}} \label{proof of Theorem 2}
We consider the conditional coverage probability according to all possible distances from a typical user to the nearest BS of the $k$th tier, i.e., take the expected value for the nearest BS distance $r$. Additionally, since the aggregated interference power of the $k$th tier $I_k$ is also a random variable, we also average over the aggregated interference power $I_k$. Namely, the conditional coverage probability is calculated as
\begin{align}
    &{\sf P}_{k}^{\sf cov}(T|{\tilde \Phi}_k(\tilde{\mathcal{A}}_k)>0,J=k)\nonumber\\
    &\!\!=\mathbb{E}\bigg[\mathbb{E}\bigg[\mathbb{P}\Big[{\rm{SINR}}_k>T\Big|{\tilde \Phi}_k(\tilde{\mathcal{A}}_k)>0,J=k,D_k=r\Big]\bigg]\bigg] \nonumber \\
    &\!\!=\mathbb{E}\bigg[\mathbb{E}\bigg[\mathbb{P}\Big[H_{k,1}\!>\!r^{\alpha_k}T (I_k+\hat{\sigma}_k^2)\Big|{\tilde \Phi}_k(\tilde{\mathcal{A}}_k)\!>\!0,J=k,D_k=r\Big]\bigg]\bigg], \label{eq:conditional coverage probability of annulus}
\end{align}
where $D_k=\lVert{\bf \tilde s}_{k,1}-\tilde{\mathbf{u}}_1\rVert$ is the distance from the nearest BS of the $k$th tier to the typical user. Using the CDF of the shadowed-Rician fading in \eqref{eq:CDF of SRF}, the CCDF of $H_{k,i}$ is obtained as
\begin{align}
    F_{H_{k,i}}^c(x)=1-\sum_{l=0}^{m_k-1}\frac{\zeta_k(l)}{(\beta_k-\delta_k)^{l+1}}\gamma(l+1,(\beta_k-\delta_k)x). \label{eq:CCDF of SRF}
\end{align}
The lower incomplete gamma function in \eqref{eq:CCDF of SRF} can be rewritten as
\begin{align}
    \gamma&(l+1,(\beta_k-\delta_k)x)=\Gamma(l+1)-\Gamma(l+1,(\beta_k-\delta_k)x) \nonumber \\
    &=\Gamma(l+1)\left(1-e^{-(\beta_k-\delta_k)x}\sum_{q=0}^{l}\frac{\left[(\beta_k-\delta_k)x\right]^q}{q!}\right).
\end{align}
Therefore, we can also rewrite the CCDF of $H_{k,i}$ in \eqref{eq:CCDF of SRF} as
\begin{equation}
    F_{H_{k,i}}^c(x)=1-\sum_{l=0}^{m_k-1}\frac{\zeta_k(l)\Gamma(l+1)}{(\beta_k-\delta_k)^{l+1}}\left(1-e^{-x}\sum_{q=0}^{l}\frac{x^q}{q!}\right).
\end{equation}
Using this CCDF of \eqref{eq:CCDF of SRF}, the conditional coverage probability in \eqref{eq:conditional coverage probability of annulus} can be specifically calculated as
\begin{align}
    &{\sf P}_{k}^{\sf cov}(T|{\tilde \Phi}_k(\tilde{\mathcal{A}}_k)>0,J=k) \nonumber \\
    &\!\!=\mathbb{E}\bigg[\mathbb{E}\bigg[\!F_{H_{k,i}}^c\!\left(r^{\alpha_k}T(I_k+\hat{\sigma}_k^2)\right)\biggr|{\tilde \Phi}_k(\tilde{\mathcal{A}}_k)>0, J=k,D_k=r\bigg]\bigg] \nonumber \\
    &\!\!\!\overset{(a)}{=}\mathbb{E}\Bigg[\!\int_{R_{{\rm{min}},k}}^{R_{{\rm{max}},k}}\!F_{H_{k,i}}^c\!\left(r^{\alpha_k}T(I_k+\hat{\sigma}_k^2)\right)\!f(r|{\tilde \Phi}_k(\tilde{\mathcal{A}}_k)\!>\!0,J=k){\rm{d}}r\Bigg] \nonumber \\
    &\!\!=\mathbb{E}\left[\int_{R_{{\rm{min}},k}}^{R_{{\rm{max}},k}}\left\{1-\sum_{l=0}^{m_k-1}\frac{\zeta_k(l)\Gamma(l+1)}{(\beta_k-\delta_k)^{l+1}}\Bigg(1-e^{-\nu_k(I_k+\hat{\sigma}_k^2)}\right.\right. \nonumber \\
    &\quad\times\left.\sum_{q=0}^{l}\frac{\left(\nu_k(I_k+\hat{\sigma}_k^2)\right)^q}{q!}\right)\Bigg\}f_k(r|{\tilde \Phi}_k(\tilde{\mathcal{A}}_k)>0,J=k){\rm{d}}r\Bigg] \nonumber \\
    &\!\!=\int_{R_{{\rm{min}},k}}^{R_{{\rm{max}},k}}\left[1-\sum_{l=0}^{m_k-1}\frac{\zeta_k(l)\Gamma(l+1)}{(\beta_k-\delta_k)^{l+1}}\Bigg(1-\sum_{q=0}^{l}\frac{\nu^q}{q!}\right. \nonumber \\
    &\!\quad\times\left.\mathbb{E}\Bigg[(I_k+\hat{\sigma}_k^2)^q e^{-\nu_k(I_k+\hat{\sigma}_k^2)}\Bigg]\right)\Bigg]f_k(r|{\tilde \Phi}_k(\tilde{\mathcal{A}}_k)>0,J=k){\rm{d}}r \nonumber \\
    &\!\!\!\overset{(b)}{=}\int_{R_{{\rm{min}},k}}^{R_{{\rm{max}},k}}\left[1-\sum_{l=0}^{m_k-1}\frac{\zeta_k(l)\Gamma(l+1)}{(\beta_k-\delta_k)^{l+1}}\Bigg(1-\sum_{q=0}^{l}\frac{\nu_k^q}{q!}(-1)^q\right. \nonumber \\
    &\;\times\frac{{\rm{d}}^q\mathcal{L}_{I_k+\hat{\sigma}_k^2|{\tilde \Phi}_k(\tilde{\mathcal{A}}_k)>0}(s)}{{\rm{d}}s^q}\Bigg|_{s=\nu_k}\Bigg)\Bigg]f_k(r|{\tilde \Phi}_k(\tilde{\mathcal{A}}_k)>0,J=k){\rm{d}}r, \label{eq:conditional coverage probability in detail}
\end{align}
where (a) is because the average is taken for $r$, and (b) comes from the derivative property of the Laplace transform with $\mathbb{E}\left[X^k e^{-sX}\right]=(-1)^k\frac{{\rm{d}}\mathcal{L}_X(s)}{{\rm{d}}s^k}$ and the Laplace transform of the aggregated interference power plus normalized noise variance is obtained in Lemma \ref{Laplace of interference power plus noise}.

\section{Proof of Theorem \ref{theo:approximation of coverage}} \label{Proof of Theorem 3}
We start with the lower and upper bounds of the lower incomplete gamma function $\gamma(l+1,(\beta_k-\delta_k)x)$ as
\begin{align}
    &\left[1-\exp\left(-\left(\Gamma(l+2)\right)^{-\frac{1}{l+1}}(\beta_k-\delta_k)x\right)\right]^{l+1} \nonumber \\
    &\qquad\quad\le\frac{1}{\Gamma(l+1)}\gamma(l+1,(\beta_k-\delta_k)x)\nonumber\\
    &\qquad\qquad\qquad\le\left[1-\exp\left(-(\beta_k-\delta_k)x\right)\right]^{l+1}.
\end{align}
These lower and upper bounds form another inequality as
\begin{align}
    &\left[1-\exp\left(-\left(\Gamma(l+2)\right)^{-\frac{1}{l+1}}(\beta_k-\delta_k)x\right)\right]^{l+1} \nonumber \\
    &\qquad\quad\le\left[1-\exp\left(-\kappa(\beta_k-\delta_k)x\right)\right]^{l+1}\nonumber\\
    &\qquad\qquad\qquad\le \left[1-\exp\left(-(\beta_k-\delta_k)x\right)\right]^{l+1}.
\end{align}
Depending on $\kappa$, the lower incomplete gamma function can be approximated as
\begin{equation}
    \!\!\!\gamma(l+1,(\beta_k-\delta_k)x)\!\approx\!\Gamma(l+1)\left[1-\exp\left(-\kappa(\beta_k-\delta_k)x\right)\right]^{l+1}.
\end{equation}
When $l=0$, $\kappa=1$ and the equality holds. However, when $l\ge 1$, $\kappa$ must be tuned to be as similar as possible.

Therefore, again using the CCDF of \eqref{eq:CCDF of SRF}, the conditional coverage probability in \eqref{eq:conditional coverage probability of annulus} can be calculated as
\begin{align}
    &{\sf P}_{k}^{\sf cov}(T|{\tilde \Phi}_k(\tilde{\mathcal{A}}_k)>0,J=k) \nonumber \\
    &\approx\mathbb{E}\left[\int_{R_{{\rm{min}},k}}^{R_{{\rm{max}},k}}\left\{1-\sum_{l=0}^{m_k-1}\frac{\zeta_k(l)\Gamma(l+1)}{(\beta_k-\delta_k)^{l+1}}\times\right.\right. \nonumber \\
    &\left[1-\exp\left(-\rho_k(I_k+\hat{\sigma}_k^2)\right)\right]^{l+1}\Bigg\}f_k(r|{\tilde \Phi}_k(\tilde{\mathcal{A}}_k)>0,J=k){\rm{d}}r\Bigg] \nonumber \\
    &\overset{(a)}{=}\int_{R_{{\rm{min}},k}}^{R_{{\rm{max}},k}}\left\{1-\sum_{l=0}^{m_k-1}\frac{\zeta_k(l)\Gamma(l+1)}{(\beta_k-\delta_k)^{l+1}}\right. \nonumber \\
    &\qquad\times\left.\mathbb{E}\left[\sum_{q=0}^{l+1}{l+1 \choose q}(-1)^q\exp\left(-\rho_k q(I_k+\hat{\sigma}_k^2)\right)\right]\right\} \nonumber \\
    &\qquad\qquad\qquad\qquad\qquad\times f_k(r|{\tilde \Phi}_k(\tilde{\mathcal{A}}_k)>0,J=k){\rm{d}}r \nonumber \\
    &\overset{(b)}{=}\int_{R_{{\rm{min}},k}}^{R_{{\rm{max}},k}}\Bigg[1-\sum_{l=0}^{m_k-1}\frac{\zeta_k(l)\Gamma(l+1)}{(\beta_k-\delta_k)^{l+1}}\sum_{q=0}^{l+1}{l+1 \choose q}(-1)^q \nonumber \\
    &\;\times\mathcal{L}_{I_k+\hat{\sigma}_k^2|{\tilde \Phi}_k(\tilde{\mathcal{A}}_k)>0}(\rho_k q)\Bigg]f_k(r|{\tilde \Phi}_k(\tilde{\mathcal{A}}_k)>0,J=k){\rm{d}}r, \label{eq:conditional coverage probability in detail}
\end{align}
where (a) is the result of binomial expansion, and (b) comes from the Laplace transform of the aggregated interference power plus normalized noise variance obtained in Lemma \ref{Laplace of interference power plus noise}.


\bibliographystyle{IEEEtran}
\bibliography{refs_all}

\begin{thebibliography}{10}
\providecommand{\url}[1]{#1}
\csname url@samestyle\endcsname
\providecommand{\newblock}{\relax}
\providecommand{\bibinfo}[2]{#2}
\providecommand{\BIBentrySTDinterwordspacing}{\spaceskip=0pt\relax}
\providecommand{\BIBentryALTinterwordstretchfactor}{4}
\providecommand{\BIBentryALTinterwordspacing}{\spaceskip=\fontdimen2\font plus
\BIBentryALTinterwordstretchfactor\fontdimen3\font minus
  \fontdimen4\font\relax}
\providecommand{\BIBforeignlanguage}[2]{{%
\expandafter\ifx\csname l@#1\endcsname\relax
\typeout{** WARNING: IEEEtran.bst: No hyphenation pattern has been}%
\typeout{** loaded for the language `#1'. Using the pattern for}%
\typeout{** the default language instead.}%
\else
\language=\csname l@#1\endcsname
\fi
#2}}
\providecommand{\BIBdecl}{\relax}
\BIBdecl

\bibitem{del2019technical}
I.~Del~Portillo, B.~G. Cameron, and E.~F. Crawley, ``A technical comparison of
  three low earth orbit satellite constellation systems to provide global
  broadband,'' \emph{Acta Astronaut.}, vol. 159, pp. 123--135, Jun. 2019.

\bibitem{guidotti2017satellite}
A.~Guidotti, A.~Vanelli-Coralli, M.~Caus, J.~Bas, G.~Colavolpe, T.~Foggi,
  S.~Cioni, A.~Modenini, and D.~Tarchi, ``Satellite-enabled lte systems in leo
  constellations,'' in \emph{2017 IEEE Int. Conf. Commun. Workshops (ICC
  Workshops)}.\hskip 1em plus 0.5em minus 0.4em\relax IEEE, Jul. 2017, pp.
  876--881.

\bibitem{cioni2018satellite}
S.~Cioni, R.~De~Gaudenzi, O.~D.~R. Herrero, and N.~Girault, ``On the satellite
  role in the era of 5g massive machine type communications,'' \emph{IEEE
  Netw.}, vol.~32, no.~5, pp. 54--61, Sep. 2018.

\bibitem{guidotti2019architectures}
A.~Guidotti, A.~Vanelli-Coralli, M.~Conti, S.~Andrenacci, S.~Chatzinotas,
  N.~Maturo, B.~Evans, A.~Awoseyila, A.~Ugolini, T.~Foggi \emph{et~al.},
  ``Architectures and key technical challenges for 5g systems incorporating
  satellites,'' \emph{IEEE Trans. Veh. Technol.}, vol.~68, no.~3, pp.
  2624--2639, Mar. 2019.

\bibitem{lin2021path}
X.~Lin, S.~Cioni, G.~Charbit, N.~Chuberre, S.~Hellsten, and J.-F. Boutillon,
  ``On the path to 6g: Embracing the next wave of low earth orbit satellite
  access,'' \emph{IEEE Commun. Mag.}, vol.~59, no.~12, pp. 36--42, Dec. 2021.

\bibitem{yang20196g}
P.~Yang, Y.~Xiao, M.~Xiao, and S.~Li, ``6g wireless communications: Vision and
  potential techniques,'' \emph{IEEE Netw.}, vol.~33, no.~4, pp. 70--75, Jul.
  2019.

\bibitem{zhang20196g}
Z.~Zhang, Y.~Xiao, Z.~Ma, M.~Xiao, Z.~Ding, X.~Lei, G.~K. Karagiannidis, and
  P.~Fan, ``6g wireless networks: Vision, requirements, architecture, and key
  technologies,'' \emph{IEEE Veh. Technol. Mag.}, vol.~14, no.~3, pp. 28--41,
  Sep. 2019.

\bibitem{qian2020integrated}
Y.~Qian, ``Integrated terrestrial-satellite communication networks and
  services,'' \emph{IEEE Wirel. Commun.}, vol.~27, no.~6, pp. 2--3, Dec. 2020.

\bibitem{qian2021multi}
L.~Qian, P.~Yang, Y.~L. Guan, Z.~Liu, Y.~Xiao, K.~Jiang, and M.~Xiao,
  ``Multi-dimensional polarized modulation for land mobile satellite
  communications,'' \emph{IEEE Trans. Cogn. Commun. Netw.}, vol.~7, no.~2, pp.
  383--397, Jun. 2021.

\bibitem{5gamericas2022NTN}
``5g and non-terrestrial networks,'' 5G americas, White paper, Feb. 2022,
  available:
  https://www.5gamericas.org/5g-americas-releases-update-on-5g-non-terrestrial-networks/.

\bibitem{al2022next}
B.~Al~Homssi, A.~Al-Hourani, K.~Wang, P.~Conder, S.~Kandeepan, J.~Choi,
  B.~Allen, and B.~Moores, ``Next generation mega satellite networks for access
  equality: Opportunities, challenges, and performance,'' \emph{IEEE Commun.
  Mag.}, vol.~60, no.~4, pp. 18--24, Apr. 2022.

\bibitem{baccelli2006aloha}
F.~Baccelli, B.~Blaszczyszyn, and P.~Muhlethaler, ``An aloha protocol for
  multihop mobile wireless networks,'' \emph{IEEE Trans. Inf. Theory}, vol.~52,
  no.~2, pp. 421--436, Jan. 2006.

\bibitem{lee2014power}
N.~Lee, X.~Lin, J.~G. Andrews, and R.~W. Heath, ``Power control for d2d
  underlaid cellular networks: Modeling, algorithms, and analysis,'' \emph{IEEE
  J. Sel. Areas Commun.}, vol.~33, no.~1, pp. 1--13, Nov. 2014.

\bibitem{al2016stochastic}
A.~Al-Hourani, S.~Kandeepan, and A.~Jamalipour, ``Stochastic geometry study on
  device-to-device communication as a disaster relief solution,'' \emph{IEEE
  Trans. Veh. Technol.}, vol.~65, no.~5, pp. 3005--3017, May 2016.

\bibitem{chun2017stochastic}
Y.~J. Chun, S.~L. Cotton, H.~S. Dhillon, A.~Ghrayeb, and M.~O. Hasna, ``A
  stochastic geometric analysis of device-to-device communications operating
  over generalized fading channels,'' \emph{IEEE Trans. Wirel. Commun.},
  vol.~16, no.~7, pp. 4151--4165, Mar. 2017.

\bibitem{andrews2011tractable}
J.~G. Andrews, F.~Baccelli, and R.~K. Ganti, ``A tractable approach to coverage
  and rate in cellular networks,'' \emph{IEEE Trans. Commun.}, vol.~59, no.~11,
  pp. 3122--3134, Oct. 2011.

\bibitem{dhillon2012modeling}
H.~S. Dhillon, R.~K. Ganti, F.~Baccelli, and J.~G. Andrews, ``Modeling and
  analysis of k-tier downlink heterogeneous cellular networks,'' \emph{IEEE J.
  Sel. Areas Commun.}, vol.~30, no.~3, pp. 550--560, Mar. 2012.

\bibitem{jo2012heterogeneous}
H.-S. Jo, Y.~J. Sang, P.~Xia, and J.~G. Andrews, ``Heterogeneous cellular
  networks with flexible cell association: A comprehensive downlink sinr
  analysis,'' \emph{IEEE Trans. Wirel. Commun.}, vol.~11, no.~10, pp.
  3484--3495, Oct. 2012.

\bibitem{huang2012analytical}
K.~Huang and J.~G. Andrews, ``An analytical framework for multicell cooperation
  via stochastic geometry and large deviations,'' \emph{IEEE Trans. Inf.
  Theory}, vol.~59, no.~4, pp. 2501--2516, Dec. 2012.

\bibitem{lee2014spectral}
N.~Lee, D.~Morales-Jimenez, A.~Lozano, and R.~W. Heath, ``Spectral efficiency
  of dynamic coordinated beamforming: A stochastic geometry approach,''
  \emph{IEEE Trans. Wirel. Commun.}, vol.~14, no.~1, pp. 230--241, Jul. 2014.

\bibitem{bai2014coverage}
T.~Bai, A.~Alkhateeb, and R.~W. Heath, ``Coverage and capacity of
  millimeter-wave cellular networks,'' \emph{IEEE Commun. Mag.}, vol.~52,
  no.~9, pp. 70--77, Sep. 2014.

\bibitem{di2015stochastic}
M.~Di~Renzo, ``Stochastic geometry modeling and analysis of multi-tier
  millimeter wave cellular networks,'' \emph{IEEE Trans. Wirel. Commun.},
  vol.~14, no.~9, pp. 5038--5057, May 2015.

\bibitem{park2016optimal}
J.~Park, N.~Lee, J.~G. Andrews, and R.~W. Heath, ``On the optimal feedback rate
  in interference-limited multi-antenna cellular systems,'' \emph{IEEE Trans.
  Wirel. Commun.}, vol.~15, no.~8, pp. 5748--5762, May 2016.

\bibitem{park2018inter}
J.~Park, J.~G. Andrews, and R.~W. Heath, ``Inter-operator base station
  coordination in spectrum-shared millimeter wave cellular networks,''
  \emph{IEEE Trans. Cogn. Commun. Netw.}, vol.~4, no.~3, pp. 513--528, Mar.
  2018.

\bibitem{wang2023resident}
R.~Wang, M.~A. Kishk, and M.-S. Alouini, ``Resident population density-inspired
  deployment of k-tier aerial cellular network,'' \emph{IEEE Trans. Wirel.
  Commun.}, Mar. 2023.

\bibitem{talgat2020stochastic}
A.~Talgat, M.~A. Kishk, and M.-S. Alouini, ``Stochastic geometry-based analysis
  of leo satellite communication systems,'' \emph{IEEE Commun. Lett.}, vol.~25,
  no.~8, pp. 2458--2462, Oct. 2020.

\bibitem{okati2020downlink}
N.~Okati, T.~Riihonen, D.~Korpi, I.~Angervuori, and R.~Wichman, ``Downlink
  coverage and rate analysis of low earth orbit satellite constellations using
  stochastic geometry,'' \emph{IEEE Trans. Commun.}, vol.~68, no.~8, pp.
  5120--5134, Apr. 2020.

\bibitem{okati2022nonhomogeneous}
N.~Okati and T.~Riihonen, ``Nonhomogeneous stochastic geometry analysis of
  massive leo communication constellations,'' \emph{IEEE Trans. Commun.},
  vol.~70, no.~3, pp. 1848--1860, Jan. 2022.

\bibitem{al2021analytic}
A.~Al-Hourani, ``An analytic approach for modeling the coverage performance of
  dense satellite networks,'' \emph{IEEE Wirel. Commun. Lett.}, vol.~10, no.~4,
  pp. 897--901, Jan. 2021.

\bibitem{al2021optimal}
------, ``Optimal satellite constellation altitude for maximal coverage,''
  \emph{IEEE Wirel. Commun. Lett.}, vol.~10, no.~7, pp. 1444--1448, Mar. 2021.

\bibitem{na2021performance}
D.-H. Na, K.-H. Park, Y.-C. Ko, and M.-S. Alouini, ``Performance analysis of
  satellite communication systems with randomly located ground users,''
  \emph{IEEE Trans. Wirel. Commun.}, vol.~21, no.~1, pp. 621--634, Jul. 2021.

\bibitem{park2022tractable}
J.~Park, J.~Choi, and N.~Lee, ``A tractable approach to coverage analysis in
  downlink satellite networks,'' \emph{IEEE Trans. Wirel. Commun.}, Aug. 2022.

\bibitem{gilhousen1991capacity}
K.~S. Gilhousen, I.~M. Jacobs, R.~Padovani, A.~J. Viterbi, L.~A. Weaver, and
  C.~E. Wheatley, ``On the capacity of a cellular cdma system,'' \emph{IEEE
  Trans. Veh. Technol.}, vol.~40, no.~2, pp. 303--312, May 1991.

\bibitem{viterbi1994other}
A.~J. Viterbi, A.~M. Viterbi, and E.~Zehavi, ``Other-cell interference in
  cellular power-controlled cdma,'' \emph{IEEE Trans. Commun.}, vol.~42, no.
  234, pp. 1501--1504, Feb.-Apr. 1994.

\bibitem{haenggi2009stochastic}
M.~Haenggi, J.~G. Andrews, F.~Baccelli, O.~Dousse, and M.~Franceschetti,
  ``Stochastic geometry and random graphs for the analysis and design of
  wireless networks,'' \emph{IEEE J. Sel. Areas Commun.}, vol.~27, no.~7, pp.
  1029--1046, Aug. 2009.

\bibitem{chiu2013stochastic}
S.~N. Chiu, D.~Stoyan, W.~S. Kendall, and J.~Mecke, \emph{Stochastic geometry
  and its applications}.\hskip 1em plus 0.5em minus 0.4em\relax John Wiley \&
  Sons, 2013.

\bibitem{wang2022ultra}
R.~Wang, M.~A. Kishk, and M.-S. Alouini, ``Ultra-dense leo satellite-based
  communication systems: A novel modeling technique,'' \emph{IEEE Commun.
  Mag.}, vol.~60, no.~4, pp. 25--31, Apr. 2022.

\bibitem{wang2022evaluating}
------, ``Evaluating the accuracy of stochastic geometry based models for leo
  satellite networks analysis,'' \emph{IEEE Commun. Lett.}, vol.~26, no.~10,
  pp. 2440--2444, Jul. 2022.

\bibitem{song2022cooperative}
Z.~Song, J.~An, G.~Pan, S.~Wang, H.~Zhang, Y.~Chen, and M.-S. Alouini,
  ``Cooperative satellite-aerial-terrestrial systems: A stochastic geometry
  model,'' \emph{IEEE Trans. Wirel. Commun.}, vol.~22, no.~1, pp. 220--236,
  Jul. 2022.

\bibitem{al2021modeling}
B.~Al~Homssi and A.~Al-Hourani, ``Modeling uplink coverage performance in
  hybrid satellite-terrestrial networks,'' \emph{IEEE Commun. Lett.}, vol.~25,
  no.~10, pp. 3239--3243, Aug. 2021.

\bibitem{park2023unified}
J.~Park, J.~Choi, N.~Lee, and F.~Baccelli, ``Unified modeling and rate coverage
  analysis for satellite-terrestrial integrated networks: Coverage extension or
  data offloading?'' \emph{arXiv preprint arXiv:2307.03343}, 2023.

\bibitem{dolph1946current}
C.~L. Dolph, ``A current distribution for broadside arrays which optimizes the
  relationship between beam width and side-lobe level,'' \emph{Proc. IRE},
  vol.~34, no.~6, pp. 335--348, Jun. 1946.

\bibitem{abdi2003new}
A.~Abdi, W.~C. Lau, M.-S. Alouini, and M.~Kaveh, ``A new simple model for land
  mobile satellite channels: First-and second-order statistics,'' \emph{IEEE
  Trans. Wirel. Commun.}, vol.~2, no.~3, pp. 519--528, May 2003.

\bibitem{kim2023coverage}
D.~Kim, J.~Park, and N.~Lee, ``Coverage analysis of dynamic coordinated
  beamforming for leo satellite downlink networks,'' \emph{arXiv preprint
  arXiv:2309.10460}, 2023.

\bibitem{FCC2023direct}
Space Exploration Holdings, LLC, SpaceX Gen2 Direct-to-Cellular System FCC
  filing SAT-MOD-20230207-00021, [Online].
  https://fcc.report/IBFS/SAT-MOD-20230207-00021.

\bibitem{westcott1972probability}
M.~Westcott, ``The probability generating functional,'' \emph{J. Aust. Math.
  Soc.}, vol.~14, no.~4, pp. 448--466, Dec. 1972.

\end{thebibliography}


 





\end{document}